\documentclass[a4paper,nofootinbib, aps, superscriptaddress]{revtex4}

\usepackage{color}
\usepackage{amsmath,amsfonts,amssymb,amsthm}
\usepackage{graphicx}
\usepackage{fullpage}
\PassOptionsToPackage{caption=false}{subfig}
\usepackage{subfig}
\usepackage{enumerate}

\usepackage{epstopdf} % to include .eps graphics files with pdfLaTeX
\usepackage{bm}  % Define \bm{} to use bold math fonts
\usepackage[pdfpagelabels,pdftex,bookmarks,breaklinks]{hyperref}
\definecolor{darkblue}{RGB}{0,0,127} % choose colors
\definecolor{darkgreen}{RGB}{0,150,0}
\hypersetup{colorlinks, linkcolor=darkblue, citecolor=darkgreen, filecolor=red, urlcolor=blue}

\usepackage{layout}

\newtheorem{theorem}{Theorem}

\newtheorem{lemma}[theorem]{Lemma}

\newtheorem{definition}[theorem]{Definition}

\renewenvironment{proof}[1][Proof]{\noindent\textbf{#1.} }{\ $\Box$}

\def\Z{\mathbb{Z}}
\def\R{\mathbb{R}}
\def\C{\mathbb{C}}

\def\e{\mathrm{e}}

\newcommand{\Eref}[1]{Eq.~(\ref{#1})}
\newcommand{\Tref}[1]{Theorem~\ref{#1}}
\newcommand{\Sref}[1]{Sec.~\ref{#1}}
\newcommand{\Fref}[1]{Fig.~\ref{#1}}

\def\th{^{\rm th}}
\def\st{^{\rm st}}
\def\nd{^{\rm nd}}

\DeclareMathOperator{\Tr}{Tr}

\newcommand{\ket}[1]{|{#1}\rangle}
\newcommand{\expect}[1]{\langle{#1}\rangle}
\newcommand{\bra}[1]{\langle{#1}|}
\newcommand{\ketbra}[2]{|{#1}\rangle\!\langle{#2}|}

\newcommand{\proj}[1]{\ketbra{#1}{#1}}

\def\pmap{\mathcal{P}}
\newcommand{\pmapbrac}[1]{\mathcal{P}_{#1}}

\def\canproj{\varpi}
\def\concatproj{\Upsilon}
\def\graph{\Lambda}
\def\pert{\varepsilon}
\def\pr{P}
\def\oconst{\tilde{\Delta}}
\def\hcurl{h}
\def\physicalqudit{code qudit }
\def\physicalspace{code space }
\def\physicalqudits{code qudits }
\def\physicalquditend{code qudit}
\def\physicalspaceend{code space}
\def\physicalquditsend{code qudits}
\def\p{c}

\begin{document}

\title{Perturbative 2-body Parent Hamiltonians for Projected Entangled Pair States}

\author{Courtney G.\ Brell}
\email{courtney.brell@gmail.com}
\affiliation{Centre for Engineered Quantum Systems, School of Physics, The University of Sydney, Sydney, NSW 2006, Australia}
\affiliation{Institut f\"ur Theoretische Physik, Leibniz Universit\"at Hannover, Appelstra\ss e 2, 30167 Hannover, Germany}
\author{Stephen D.\ Bartlett}
\affiliation{Centre for Engineered Quantum Systems, School of Physics, The University of Sydney, Sydney, NSW 2006, Australia}
\author{Andrew C.\ Doherty}
\affiliation{Centre for Engineered Quantum Systems, School of Physics, The University of Sydney, Sydney, NSW 2006, Australia}

\date{11 November 2014}

\begin{abstract}
We construct parent Hamiltonians involving only local 2-body interactions for a broad class of Projected Entangled Pair States (PEPS). Making use of perturbation gadget techniques, we define a perturbative Hamiltonian acting on the virtual PEPS space with a finite order low energy effective Hamiltonian that is a gapped, frustration-free parent Hamiltonian for an encoded version of a desired PEPS. For topologically ordered PEPS, the ground space of the low energy effective Hamiltonian is shown to be in the same phase as the desired state to all orders of perturbation theory. An encoded parent Hamiltonian for the double semion string net ground state is explicitly constructed as a concrete example.
\end{abstract}

\maketitle

%\vspace{-1.5cm}
% new page after abstract & TOC
%\clearpage
%\tableofcontents
%\clearpage

%------------------------------------------------------------------------------------------------------------%
%------------------------------------------------------------------------------------------------------------%

\section{Introduction}

Projected Entangled Pair States (PEPS) are a class of quantum states particularly well suited for describing the ground states of interacting quantum many-body systems~\cite{Affleck1988, Fannes1992, Hastings2006, Perez-Garcia2007a, Verstraete2006, Vidal2003a}. They are a form of tensor network ansatz amenable to both numerical and analytical study, and encompass many interesting classes of states. In particular, they offer exact analytical descriptions of such states as the topologically ordered ground states of quantum double models~\cite{KitaevTC97} and string-net models~\cite{LevinSN05}, as well as resources for measurement-based quantum computation such as the cluster states~\cite{Raussendorf2001} and Affleck-Kennedy-Lieb-Tasaki (AKLT) states~\cite{Affleck1988,Wei2011,Darmawan2012}, among others.

For a given PEPS (representing the state of a quantum many-body system defined on a graph), there is an associated parent Hamiltonian for which it is a ground state~\cite{Perez-Garcia2008}. For certain classes of PEPS, these Hamiltonians can be defined using only local interactions (i.e., interactions whose support lies only on qudits within some bounded size region) such that their ground states are unique. Though these interactions act only within a finite sized region, there will still generally be a large number of qudits within this region. For this reason these interactions may be challenging to implement experimentally, and it may be preferable to find an alternative parent Hamiltonian with interactions involving at most two neighbouring quantum systems (2-local interactions), whose ground state is a desired PEPS.

In this paper, we construct such a parent Hamiltonian involving only 2-local interactions for PEPS with certain properties. The strategy we use to show this is based on the perturbation gadget~\cite{Kempe2006, Oliveira2008, Bravyi2008a, Jordan2008} approach. Perturbation gadgets allow $k$-body interactions (those involving $k$ systems) to be approximated by 2-body interactions through the introduction of ancilla qudits coupled perturbatively. Applied to infinite systems, a naive perturbation gadgets approach can encounter a number of pitfalls. In particular, the resource cost of a general perturbation gadget scheme scales poorly with the system complexity, and application of the technique can lead to the energy gap scaling with the system size or the fidelity of target states~\cite{VandenNest-08,Oliveira2008}. Additionally, while interactions can be reduced to only 2-body, the nature of these interactions is in general complicated and unnatural when viewed in terms of the target model. By taking advantage of structure and tailoring the gadgets to a particular class of models, we can circumvent these difficulties. Our construction involves interactions that are natural from the point of view of the PEPS ansatz, and captures the structure of the standard PEPS parent Hamiltonian.

We present a perturbation gadget scheme that works by encoding the qudits of the model in question in a quantum code, and weakly coupling neighbouring encoded qudits. The encodings and couplings are directly inspired by the PEPS descriptions of the target ground states. As such, our scheme is specifically suited to constructing 2-local Hamiltonians whose ground space is (an encoded form of) a desired PEPS. This generalizes the ideas of Refs.~\cite{Bartlett-2006} and \cite{Brell-2011}, where similar techniques were used to reproduce encoded forms of the cluster state and the quantum double ground states, respectively, as the ground states of entirely 2-local Hamiltonians, based on their PEPS descriptions. The model studied in this paper is not precisely equivalent to those developed previously, which take advantage of structure that is not generally available for all the PEPS we discuss here.  As well as the quantum double and cluster states, we expect our construction to apply to broad classes of topologically ordered states with similar structure, such as the string-net ground states, isometric $H$-injective PEPS~\cite{Buerschaper2010a}, and $(G,\omega)$-isometric PEPS~\cite{Buerschaper2013}. In this direction, we argue that isometric MPO-injective PEPS~\cite{Sahinoglu2014} with trivial so-called generalized inverse satisfy the requirements of our construction; this class is known to include string-net ground states and $(G,\omega)$-isometric PEPS. In fact, we conjecture that our results extend to any PEPS that satisfy certain topological order conditions.

Our analysis proceeds in two parts. In the first part, we show that for a given PEPS satisfying certain criteria there exists a finite-order low energy effective Hamiltonian for our system which is a valid (gapped) parent Hamiltonian for the desired PEPS satisfying these conditions. Our perturbation analysis is based on the Schrieffer-Wolff transformation~\cite{Schrieffer1966, Bravyi2011}. In the second part of our analysis, we study the robustness of this effective Hamiltonian to contributions from higher order terms in the perturbation expansion. In doing so, we prove that the effective Hamiltonian is in the same phase as the desired parent Hamiltonian to arbitrary order. 

We make use of stability results for topologically ordered states~\cite{Bravyi2010, Bravyi2010a, Michalakis2011, Cirac2013} to prove that the ground space of our effective Hamiltonian remains in the same phase to arbitrarily high order of perturbation theory. For this reason, our results apply only to states with parent Hamiltonians satisfying the local topological quantum order conditions~\cite{Michalakis2011}. Note that these states need not be topologically ordered in the more colloquial sense, and may have non-degenerate ground spaces, for example.

As an explicit new example of our construction, we demonstrate our procedure for the double semion string-net model~\cite{LevinSN05}, which has an exact PEPS description~\cite{Gu2008,Gu-09,Buerschaper-09}.

In \Sref{s:peps} we will introduce the PEPS formalism and define the class of PEPS to which our construction applies. In \Sref{s:overview} we briefly outline our model and main results. Following this, Sections \ref{s:perturb} and \ref{s:stable} are devoted to the proofs of our results. \Sref{s:discuss} contains a discussion of our results and concluding remarks, followed by an explicit example of our construction in \Sref{s:egds}.

%------------------------------------------------------------------------------------------------------------%
%------------------------------------------------------------------------------------------------------------%

\section{Projected Entangled Pair States}\label{s:peps}

PEPS is an ansatz for describing states of many-body quantum systems. For a given PEPS satisfying certain criteria, we will exploit the structure of this ansatz in order to construct a 2-local Hamiltonian whose ground state is in the same phase as the desired PEPS.

A PEPS is typically associated with a graph or lattice $\graph$, and can be defined constructively by beginning with maximally entangled pairs of qudits of dimension $D$ on each edge $e=(i,j)$ of the graph, such that one qudit from each pair is associated with each of the sites $i$ and $j$. These qudits are conventionally called virtual qudits. A linear map $\pmap_s:\left(\C^D\right)^{\otimes \mathrm{deg}(s)}\rightarrow \C^d$ is then applied to the $\mathrm{deg}(s)$ virtual qudits at each site $s$ of the graph (with $\mathrm{deg}(s)$ the degree of $s$), mapping the combined Hilbert spaces of all the virtual qudits at $s$ to an encoded space of dimension $d$. The space $\C^d$ is often called the physical space, but we will refer to it as the \physicalspaceend, associated with an encoding of a $d$-dimensional qudit in the $\mathrm{deg}(s)$ $D$-dimensional virtual qudits. The map $\pmap_s$ is often referred to as the projection map, though it need not be a projector in the strict sense.

	This structure of a PEPS is illustrated in \Fref{F:PEPS}. In general, the projection map and the dimensions $d$ and $D$ can vary with location, but for notational simplicity we will restrict our attention to the translation-invariant case (extension to the general case is straightforward). We also take the graph $\graph$ to have a bounded coordination number, i.e.~$\mathrm{deg}(s)$ is finite.  Because each edge $e=(i,j)$ of the PEPS graph has an associated pair of virtual qudits, we define $\ket{\Phi_D(e)}\equiv\sum_{k=0}^{D-1}\ket{k}_{e,i}\ket{k}_{e,j}$ as the maximally entangled state on edge $e$, where $\ket{\cdot}_{e,i}$ refers to the state of the virtual qudit at site $i$ corresponding to edge $e$. With this in mind, we can write the PEPS state (up to normalization) as
	\begin{align}\label{e:PEPSdefp}
		\ket{\psi_{\mathrm{PEPS}}}_{\p} &=\prod_{s}\pmap_s\prod_{e}\ket{\Phi_D(e)}_v
	\end{align}
	where $e$ runs over edges of $\graph$, and $s$ runs over sites of $\graph$. Note that $\ket{\psi_{\mathrm{PEPS}}}_{\p}$ is defined on the \physicalspaceend, as opposed to the virtual qudits. We denote states and operators on the \physicalqudits and virtual qudits by subscript $\p$ and $v$, respectively. 
	
%------------------------------------------------------------------------------------------------------------%
	\begin{figure}
	\centering
	\subfloat{
	\includegraphics{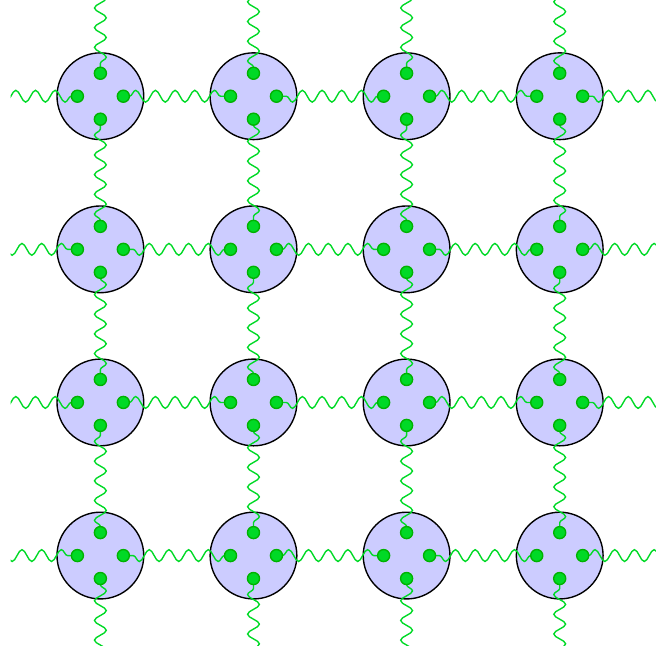}
	}
	\caption{PEPS construction on a square lattice. Virtual systems are shown in green, while \physicalqudits are shown in blue, encoded in the enclosed virtual qudits. Wavy lines run along edges of the PEPS graph, and connect virtual qudits in maximally entangled states.}
	\label{F:PEPS}
	\end{figure}
	%------------------------------------------------------------------------------------------------------------%

	A parent Hamiltonian of a PEPS (or more generally any state or space) is a gapped local Hamiltonian that has the desired state as its ground state. It is often conventional to also require that a parent Hamiltonian is frustration-free, but we will still refer to a frustrated Hamiltonian as a parent Hamiltonian so long as it is local and gapped. For any given PEPS, there is a special class of parent Hamiltonians, \emph{canonical} parent Hamiltonians~\cite{Perez-Garcia2008}, defined as follows.
	
	A canonical parent Hamiltonian is specified by a set of regions $\{R\}$ on the PEPS lattice, where $\{R\}$ must contain a region of a large enough size around each site of the lattice. We call the largest required region size $r_*$.  This size $r_*$ generally depends on the details of the PEPS under consideration, but in the cases we consider it can always be taken to be finite (see \cite{Perez-Garcia2008} for details). For each $R$, we define the projector $\canproj_R$ onto the support of $\rho_R = \Tr_{\setminus R}\proj{\psi_{\mathrm{PEPS}}}$, the reduced state in region $R$. The associated canonical parent Hamiltonian is then
	\begin{align}\label{e:canparentHamp}
		H_{\mathrm{can},\p} = -\sum_R \canproj_{R}
	\end{align}
This Hamiltonian will have the desired PEPS as a frustration-free ground state. The Hamiltonian (\ref{e:canparentHamp}) acts on the \physicalqudits of the model (as denoted by the subscript $\p$), and the virtual qudits are seen only as a mathematical tool used in the definition of the PEPS.
	
	\subsubsection{Virtual qudits and \physicalquditsend}	
	
	A PEPS is a state defined on the \physicalqudits as in Eq.~(\ref{e:PEPSdefp}). However, we will consider the virtual qudits to be those that are manipulated in the laboratory and we regard the \physicalqudits to simply be encoded within a subspace of these virtual systems. This is the sense in which we will recover an encoded form of the PEPS.
	
	Generically in perturbation gadget approaches, ancilla qudits are required to mediate effective many-body correlations. In standard approaches, these ancillae are introduced beside the original qudits of the model being considered and simply mediate the interactions between the model qudits. However, in our construction each qudit of the desired model is encoded into several ancilla qudits, and so the additional qudits are integrated into the structure of the model itself. In this way we have aligned the structure of our perturbation gadgets with that of PEPS to reproduce these states more naturally.
	
	 For clarity, for the bulk of our analysis the interactions we use will not generally be 2-local on these virtual systems, but instead be 2-local on the collection of virtual qudits that encodes each \physicalquditend. We call this collection a \emph{code gadget}. We stress this distinction between 2-locality of an interaction with respect to the code gadgets as opposed to the virtual qudits as it is a departure from previous similar work~\cite{Bartlett-2006,Brell-2011}.  We discuss in~\Sref{s:reducecomplex} how to subsequently construct Hamiltonians whose interactions involve at most 2 virtual qudits.

	\subsection{Types of PEPS}\label{s:quasiinj}

	Because we take a perturbative approach in this work, we will construct a parent Hamiltonian for a state within the same phase as a given PEPS, such that the ground state of our model can be made arbitrarily close to the PEPS by taking the perturbation parameter to be small enough.  For this to be a sensible approach, we require parent Hamiltonians for the PEPS under consideration to be gapped and stable (in an appropriate sense) with respect to small perturbations. This criterion is formalized as \emph{topological order}, and will be discussed in \Sref{s:topopeps}. As well as the topological order condition, our procedure will require the PEPS we treat to possess additional structure as compared to the most general definition of PEPS. We are interested in the broadest such structure that will allow us to demonstrate our result.
	
	Apart from topologically ordered PEPS, there are two main subclasses of PEPS that we will need to consider: \emph{isometric} PEPS and \emph{quasi-injective} PEPS.  Isometric PEPS are a natural and important subclass of PEPS, and are renormalization fixed points~\cite{Schuch2010a}. Quasi-injectivity is the least natural of the classes we consider and is mainly a technical tool required in our analysis. We argue that it is a generalization of several known classes of PEPS such as injective PEPS~\cite{Perez-Garcia2008} and $G$-injective PEPS~\cite{Schuch2010}, which have properties that make them amenable to our construction.
	
	\subsubsection{Isometric PEPS}

	\begin{definition}[Isometric PEPS~\cite{Schuch2010a, Schuch2010}]  
	A PEPS is isometric if the projection maps $\pmap_s$ are isometries.
	\end{definition}
	Notably, for isometric PEPS $P_s\equiv\pmapbrac{s}^{\dagger}\pmap_s$ is a (Hermitian, idempotent) projector acting on the virtual qudit space.
	
	Apart from being renormalization fixed-points, isometric PEPS also give a simpler form for the parent Hamiltonian than the general case~\cite{Schuch2010}. In this work, we will require all PEPS we treat to be isometric, as we will make use of the additional structure of their parent Hamiltonians in our analysis. For isometric PEPS, we can write a canonical parent Hamiltonian on the virtual space as
	\begin{align}
		H_{\mathrm{can},v} &= \pmapbrac{\graph}^{\dagger}H_{\mathrm{can},\p}\pmap_\graph\\
		&=-\sum_R\pmapbrac{\graph}^{\dagger}\canproj_R\pmap_\graph\label{e:canparentHamv}
	\end{align}
	where $\pmap_\graph=\prod_{s\in \graph}\pmap_s$. Explicitly, this Hamiltonian will be be a parent Hamiltonian for the encoded PEPS state
	\begin{align}\label{e:PEPSdefv}
		\ket{\psi_{\mathrm{PEPS}}}_v &=\prod_{s}P_s\prod_{e}\ket{\Phi_D(e)}_v
	\end{align}
	defined on the virtual space.
		
	\subsubsection{Quasi-injective PEPS}
	We will define a class of PEPS that we call \emph{quasi-injective} PEPS, inspired by several known classes of PEPS. The most fundamental of these known classes is injective PEPS~\cite{Perez-Garcia2008}, which are technically defined as those PEPS whose projection maps have left inverses. Injectivity has important consequences for properties of the parent Hamiltonian, and in particular injective PEPS can be shown to be unique ground states of their canonical parent Hamiltonians, which can be defined to be $2$-local~\cite{Perez-Garcia2008}. Broader classes of PEPS with similar structure have been proposed, such as $G$-injective PEPS~\cite{Schuch2010} for finite groups $G$, $(G,\omega)$-injective PEPS~\cite{Buerschaper2013} for a finite group $G$ and 3-cocycle $\omega$, and $H$-injective PEPS~\cite{Buerschaper2010a} for finite-dimensional $C^*$ Hopf algebras $H$.  Recently, a notion of injectivity based on the use of a projection matrix product operator (MPO) that includes many (perhaps all) of these previous classes has been developed~\cite{Sahinoglu2014}, known as MPO-injectivity.  MPO-injective PEPS can describe a large class of topologically ordered states including the ground states of string-net models~\cite{LevinSN05}. \color{black} In contrast to injective PEPS, which represent unique ground states of local Hamiltonians, these other classes typically represent the ground states of topologically ordered systems that would generally have degenerate ground spaces.
	
	For all of these classes of PEPS (injective, $G$-injective, etc.), the canonical parent Hamiltonians can be shown to have additional structure that is not present in the general case. In particular, isometric PEPS that are also injective or $G$-injective have canonical parent Hamiltonians (\ref{e:canparentHamp}) whose terms take a particularly simple form:
		\begin{align}\label{e:injectHam}
			\canproj_R= \pmap_R\cdot\left(\prod_{e\in R}\proj{\Phi_D(e)}_v\right)\cdot\pmapbrac{R}^{\dagger}
		\end{align}
		Importantly, the Hamiltonian (\ref{e:canparentHamp}) can be chosen to be both local and gapped for these PEPS~\cite{Perez-Garcia2008,Schuch2010a,Schuch2010}.
		
		Our construction will, among other things, require the PEPS under consideration to have canonical parent Hamiltonians that are local, gapped, and whose terms take the form of Eq.~(\ref{e:injectHam}).  Injective or $G$-injective isometric PEPS have these properties, but we wish to be as general as possible.  We will therefore define the class of quasi-injective PEPS to be those that satisfy the loosest such conditions that are sufficient to prove our main results.  We believe that, for isometric PEPS, our definition of quasi-injectivity generalizes the known classes of PEPS mentioned earlier (injective, $G$-injective, $(G,\omega)$-injective, $H$-injective, and MPO-injective). Loosely speaking, the conditions we impose require that the PEPS is stabilized by a set of operators $\concatproj_{\{R_{P_i},R_{E_i}\}}$ defined below, in the sense that the PEPS is an eigenstate of each $\concatproj_{\{R_{P_i},R_{E_i}\}}$ corresponding to its highest eigenvalue. Explicitly, these (Hermitian) operators take the form
			\begin{align}\label{e:defupsilons}
				\concatproj_{\{R_{P_i},R_{E_i}\}}&= \frac{1}{2}P_{\bigcup_j\{R_{P_j},R_{E_j}\}}\left(\prod_{i}P_{R_{P_i}}\left(\prod_{e\in {R_{E_i}}}\proj{\Phi_D(e)}_v\right)\right)P_{\bigcup_j\{R_{P_j},R_{E_j}\}}+\mathrm{h.c.}
			\end{align}
			where the $R$ are connected regions of the graph and $P_R=\prod_{s\in R}\pmapbrac{s}^{\dagger}\pmap_s$. The set ${\{R_{P_i},R_{E_i}\}}=\{R_{P_1},R_{P_2},\ldots,R_{E_1},R_{E_2},\ldots\}$ is a set of regions, with $\bigcup_j\{R_{P_j},R_{E_j}\}$ their union.
			
			PEPS is a tensor network ansatz in the sense that the projection map can be defined by a tensor with indices corresponding to each virtual and \physicalquditend. In this picture, the operators $\concatproj_{\{R_{P_i},R_{E_i}\}}$ can be thought of as contractions of the tensors defining the PEPS projector in various ways. This is because each $P_s$ can be thought of as a tensor with input and output indices for each virtual qudit, and $\proj{\Phi_D(e)}_v$ acts to contract the relevant indices of these tensors on the sites $e$ connects. In particular, $\concatproj_{\{\emptyset,R\}}=P_R\left(\prod_{e\in R}\proj{\Phi_D(e)}_v\right)P_{R}=\pmapbrac{R}^{\dagger}\canproj_R\pmap_{R}$ corresponds to the contraction of all of the pairs of indices of $P_R$ corresponding to edges in $R$. 
		
		Given these $\concatproj_{\{R_{P_i},R_{E_i}\}}$ operators, we can explicitly define the quasi-injectivity condition as follows:
		\begin{definition}[Quasi-injective PEPS]
			 An isometric PEPS is quasi-injective if
			\begin{align}
					%\left(\concatproj_{\{R_{P_i},R_{E_i}\}}\right)^2&=\kappa_{\{R_{P_i},R_{E_i}\}}\concatproj_{\{R_{P_i},R_{E_i}\}}\,,\label{e:quasidef1}\\
					\concatproj_{\{R_{P_i},R_{E_i}\}}\ket{\psi_{\mathrm{PEPS}}}_v&=\eta_{\{R_{P_i},R_{E_i}\}} \ket{\psi_{\mathrm{PEPS}}}_v \,,\label{e:quasidef2}
			\end{align}
			for $\eta_{\{R_{P_i},R_{E_i}\}}$ the largest eigenvalue of $\concatproj_{\{R_{P_i},R_{E_i}\}}$.
		\end{definition}
			
			Most significantly, isometric quasi-injective PEPS have the property that
			\begin{align}\label{e:SHamil}
				H_{\mathrm{par},v}&=-\sum_{\{R_{P_i},R_{E_i}\}}c_{\{R_{P_i},R_{E_i}\}} \concatproj_{\{R_{P_i},R_{E_i}\}} \,,
			\end{align}
			is a valid frustration-free parent Hamiltonian for any choice of $c_{\{R_{P_i},R_{E_i}\}}>0$. This is a direct consequence of the condition (\ref{e:quasidef2}). It is also clear from the fact that $H_{\mathrm{par},v}$ contains all terms in the canonical parent Hamiltonian that if (\ref{e:canparentHamp}) is gapped, the Hamiltonian (\ref{e:SHamil}) is also gapped.
			
	%This is true by the following argument: Consider a Hamiltonian as a (positive) sum of projectors by adding a constant energy shift. We assume that the ground state is frustration free and gapped. Now, by adding further projectors which are also frustration free, we will not affect the ground state energy, and we can only increase any other energy eigenstate (as projectors are positive semi-definite). Therefore the gap cannot decrease when adding these extra projectors.
	
	We emphasise that the notion of quasi-injectivity is designed to be the loosest notion required for our results to hold, and it is expected (in the isometric case) to encompass the broad class of PEPS listed above, including injective, $G$-injective, $(G,\omega)$-injective, $H$-injective, and MPO-injective PEPS.  To illustrate this, we provide a proof sketch that all MPO-injective PEPS for which the `generalized inverse' (defined in Ref.~\cite{Sahinoglu2014}) is the identity map are quasi-injective, noting that this class includes injective, $G$-injective, $(G,\omega)$-injective and string-net PEPS.  For isometric PEPS, as we consider in this paper, it is believed that all MPO-injective PEPS have a generalized inverse equal to the identity.  We also believe that quasi-injectivity captures the relevant features of higher dimensional analogues of these classes, such as projected entangled pair operator (PEPO)-injective PEPS that are the natural extension of MPO-injective PEPS.  We leave the development of a formal proof of the relationship between quasi-injectivity and other forms of injectivity to future work.  The close relationship between MPO-injective PEPS and topologically ordered PEPS in 2 dimensions may also suggest a close relationship between quasi-injective PEPS and topologically ordered PEPS in general.
	
    \textit{Proof sketch:}  Consider the operator $\concatproj_{\{R_{P_i},R_{E_i}\}}$ acting on an MPO-injective PEPS.  Consider each projection factor of $\concatproj_{\{R_{P_i},R_{E_i}\}}$ as defined in Eq.~\eqref{e:defupsilons} to be applied sequentially.  As we apply projectors on maximally-entangled states associated with some set of bonds on the lattice, this acts to block sites.  This is because for MPO-injective PEPS, the generalised inverse tensor can be considered to be a blocking operation.  For those MPO-injective PEPS with trivial generalized inverse, this simply corresponds to a contraction of the relevant bond, which in our context is implemented by the projection onto the maximally entangled state.  Applying PEPS projectors on some set of sites then removes these sites from any blocks.  The pull-through condition of MPO-injective PEPS ensures that blocks with such sites removed indeed remain blocks. Note that $\concatproj_{\{R_{P_i},R_{E_i}\}}$ is composed of a sequence of such bond projections that block sites followed by PEPS projectors that remove sites from blocks.  At the conclusion of this sequence, the state is described by blocks of sites (not necessarily geometrically local); however, applying the PEPS projector at all sites restores the original PEPS state, thereby guaranteeing that the PEPS is stabilized by all $\concatproj_{\{R_{P_i},R_{E_i}\}}$ as required by quasi-injectivity.  

	\subsubsection{Topological Order}\label{s:topopeps}

	The quasi-injective and isometric conditions discussed above are specific to the PEPS framework. In contrast, the topological order condition applies more generally to frustration-free, gapped, local Hamiltonian systems.  Systems with topological order have inherent stability to quasi-local perturbations (defined below). Since we will be using a perturbative approach, we must consider the effect of high order corrections in the perturbation expansion, and topological stability results will be crucial to establishing the robustness of our results to these corrections. In our context, the topological order condition will be applied to the family of canonical parent Hamiltonians (\ref{e:canparentHamp}) for a given PEPS.
	
	The following results have been developed in a sequence of works~\cite{Bravyi2010, Bravyi2010a, Michalakis2011} on the definition and stability of topologically ordered systems (see also related work on stability of tensor network states~\cite{Szehr2014,Cirac2013}). It is not our intention to provide a complete discussion of topological order, and we will we will simply paraphrase the relevant definitions and results here. We neglect several details (in particular, we restrict our discussion to infinite length scales); interested readers should consult Ref.~\cite{Michalakis2011} for a more thorough treatment and for technical details of these conditions.
	
	\begin{definition}[Local-TQO]
		Consider a gapped Hamiltonian $H_0=\sum_{u\in \graph} L_u$ for some $L_u$ supported in local regions around site $u$ of a graph $\graph$.  Let $P_0(R)$ be the projector to the ground space of the restricted Hamiltonian $H_R=\sum_{u\in R}L_u$ for a region $R\subseteq \graph$.  The system is said to obey \emph{Local-TQO} iff for all operators $X_R$ acting on a finite region $R$, there exists a superpolynomially decaying function $f(r)$ such that
		\begin{align}
			\left\|P_0(R_r)X_R P_0(R_r)-\frac{\Tr P_0(R_r)X_R}{\Tr P_0(R_r)}P_0(R_r)\right\|\leq \|X_R\| f(r)
		\end{align}
		for all finite $r$, where $R_r$ is a region enclosing all points within distance $r$ from $R$ (including $R$ itself).
	\end{definition}	
	\begin{definition}[Local-Gap]
		Given a gapped, local Hamiltonian $H_0$, we say that it obeys the \emph{Local-Gap} condition iff for each $R\subseteq\graph$ and $r\geq 0$, we have that $H_{R_r}$ has gap at least $g(r)$ for $g$ a function decaying at most polynomially in $r$.
	\end{definition}
	
	The Local-TQO condition formalises the colloquial definition of a topologically ordered system as one whose ground states cannot be distinguished by local operations. The Local-Gap condition is required to prove the topological stability results below. %For a detailed discussion see Ref.~\cite{Michalakis2011}.
	
	If at least one canonical parent Hamiltonian for the PEPS satisfies Local-TQO and Local-Gap, we say that the PEPS is topologically ordered. One might be concerned that in general some special choices of canonical parent Hamiltonian will satisfy these conditions while the rest will not. It can easily be seen that if one canonical parent Hamiltonian defined by a set of regions $\{R\}$ satisfies the topological order conditions, then the family of canonical parent Hamiltonians defined by region sets $\{R'\}\supseteq\{R\}$ are also topologically ordered. Thus for large enough sets of regions the Local-TQO and Local-Gap conditions are universal properties of a PEPS, rather than properties of a specific canonical parent Hamiltonian.

We will also require a notion of quasi-locality for operators.  An operator $X$ will be called $(J,\mu)$-quasi-local iff it has a local decomposition $X = \sum_{s\in \graph}\sum_r X_{s,r}$ for $X_{s,r}$ with support only within radius $r$ of site $s$, and $\|X_{s,r}\|\leq J\mu^{r}$ for some $\mu<1$.  

	Given a PEPS that is topologically ordered, we can make use of the topological stability theorem:	
	\begin{theorem}[Topological Stability~\cite{Michalakis2011}]\label{t:topologicalstability}
		Given a frustration-free Hamiltonian $H_0$ with $O(1)$ gap, satisfying the Local-TQO and Local-Gap conditions, there exists $\pert>0$ such that $H=H_0+\pert V$ has spectral gap $O(1)$ for quasi-local $V$.
	\end{theorem}
	
	Note that systems satisfying both Local-TQO and Local-Gap conditions need not have degenerate ground spaces as one might expect for a conventional notion of a topologically ordered system. In particular, canonical parent Hamiltonians of injective PEPS as well as $G$-injective PEPS, etc.~are topologically ordered by this definition. It seems natural to conjecture that the appropriate definition of quasi-injectivity is really equivalent to (or at least implied by) the topological order conditions. Unfortunately it is unclear to us how to prove this conjecture.

%------------------------------------------------------------------------------------------------------------%
%------------------------------------------------------------------------------------------------------------%

\section{Overview of results}\label{s:overview}

	In this section we give an outline of our method and results. Given a suitable PEPS, our goal is to construct a quantum spin model with 2-body interactions that is a parent Hamiltonian for a state within the same phase as this PEPS. We will first describe the form of the Hamiltonian by which we achieve this, before stating our main theorems. Sections \ref{s:perturb} and \ref{s:stable} will be devoted to proving these theorems.

	\subsection{Construction}\label{s:model}

	Our strategy will be to use a perturbative Hamiltonian to simulate the different elements of the PEPS construction. In contrast to a conventional PEPS parent Hamiltonian, our model acts on the virtual qudit space, as opposed to the \physicalqudit space. The unperturbed dynamics of our model will be such that the ground space of our system is an encoded form of the relevant PEPS \physicalspaceend.

	Our main tool is the code gadget~\cite{Brell-2011}, which can be understood as the collection of virtual qudits at each site of the PEPS lattice, together with a Hamiltonian whose ground space is the desired code space (i.e.~an encoding of a physical qudit in the PEPS language). From this point on all operators act on the virtual qudits unless otherwise stated, and so we will suppress $v$ subscripts on operators and states. The encoding of the $d$-dimensional \physicalqudit in the virtual space is given (up to 1-local unitaries) by the projection map $\pmap_s$. The simplest way to achieve this encoding in the ground space of a code gadget is by using the Hamiltonian
	\begin{align}
		Q_s&=1-\pr_s
	\end{align}
	with $\pr_s\equiv \pmapbrac{s}^{\dagger}\pmap_s$ the projector to the PEPS \physicalspace (for an isometric PEPS). Each of these code gadgets therefore corresponds to a single \physicalquditend.
	
	Note that the $Q_s$ act on a single code gadget, but this corresponds to $\mathrm{deg}(s)$ virtual qudits. Thus, in terms of the virtual qudits of the model, this is a $\mathrm{deg}(s)$-body interaction. We will generally analyse this model as written, and in \Sref{s:reducecomplex} discuss how to reduce the interactions from $\mathrm{deg}(s)$-body to 2-body on the virtual qudits if required.

	We couple the code gadgets perturbatively according to the structure of the PEPS lattice $\graph$, and this coupling will mediate the correlations present in the PEPS. For each edge $e$ of the PEPS lattice, we define a coupling term
	\begin{align}
		M_e &= \proj{\Phi_D(e)}
	\end{align}
	with $M$ the projector to the maximally entangled state of the relevant virtual space dimension $D$.
	
	The Hamiltonian of our system is then given by
	\begin{align}\label{e:fullHami}
		H&=\sum_{s\in \graph} Q_s-\pert\sum_e M_e
	\end{align}
	where $\pert \ll 1$. This is a 2-body Hamiltonian (considering each code gadget to be a single particle), and we will show that it is a valid parent Hamiltonian for a state in the same phase as the desired isometric, quasi-injective, topologically ordered PEPS.
	
	A simple example of this construction is analysed in \Sref{s:egds}. The analysis uses a simplified formalism and is much more accessible than the main technical sections; some readers may wish to read it before tackling the technical issues that are required to treat the general case.

\subsection{Results}

	Our analysis of the model described above proceeds in several stages. The main idea is to compute a perturbation expansion in $\pert$ for a low-energy effective Hamiltonian of the system and analyse its properties. To achieve this, we use the global Schrieffer-Wolff perturbation method~\cite{Schrieffer1966,Bravyi2011} as described in \Sref{s:perturb}. We find the following result:
	
	\begin{theorem}\label{t:nthorder}
		Given an isometric, quasi-injective PEPS with gapped canonical parent Hamiltonian, there exists finite $n_*$ such that the global Schrieffer-Wolff effective Hamiltonian for our model to order $n_*$ in the perturbation parameter $\pert$ is a frustration-free parent Hamiltonian for the PEPS and has gap $O(\pert^{n_*})$ for sufficiently small $\pert>0$.
	\end{theorem}

	This theorem is proved in \Sref{s:perturb} and is the most crucial part of our analysis. 
	
	Beyond this result, we would like to demonstrate that the full Hamiltonian (\ref{e:fullHami}) is gapped and has ground state in the same phase as the desired PEPS.  In this direction, we analyse the stability of the gap of the $n_*\th$ order effective Hamiltonian to the addition of higher order terms in the perturbation expansion.  Such stability would guarantee that the $n_*\th$ order effective Hamiltonian is adiabatically connected to the full effective Hamiltonian at arbitrary order in perturbation theory, and so their ground states are in the same phase. 
	
	In order to guarantee any kind of stability against the additional contribution from higher order terms in the perturbation expansion, we appeal to known results for topologically ordered systems~\cite{Michalakis2011,Bravyi2010,Cirac2013}. If a given PEPS is topologically ordered as well as isometric and quasi-injective, then we can use Theorem (\ref{t:nthorder}) together with the local Schrieffer-Wolff perturbation method~\cite{Datta1996,Bravyi2011} to demonstrate the following theorem:

	\begin{theorem}\label{t:stability}
		Given an isometric, quasi-injective, topologically ordered PEPS, there exists $\pert>0$ such that the effective Hamiltonian for our model to any order $k>n_*$ is in the same phase as the $n_*\th$ order effective Hamiltonian.
	\end{theorem}

	Here we define two ground states of gapped quasi-local Hamiltonians $H_1$ and $H_2$ to be in the same phase iff $H_1$ can be connected to $H_2$ by a quasi-local adiabatic evolution that does not close the spectral gap. Theorems \ref{t:nthorder} and \ref{t:stability} straightforwardly imply the following theorem, which is the main result of this paper.
	
	\begin{theorem}
		There exists $\pert>0$ such that the low-energy effective Hamiltonian corresponding to the system (\ref{e:fullHami}) is a gapped parent Hamiltonian for a state in the same phase as the quasi-injective, isometric, topologically ordered PEPS under consideration to any order $k>n_*$ of perturbation theory.
	\end{theorem}

	Our proof of \Tref{t:stability} is given in \Sref{s:stable} and involves two stages. We first transform the effective Hamiltonian derived from the global Schrieffer-Wolff method to one defined by the related local Schreiffer-Wolff method~\cite{Datta1996,Bravyi2011}. Although the global SW expansion has structure which allows for the proof of Theorem \ref{t:nthorder}, the higher-order terms in its expansion cannot easily be bounded. Conversely, the local SW expansion has explicit locality properties that allow us to analyse higher-order terms, but does not allow for a direct proof of Theorem \ref{t:nthorder}. By transforming between these two effective Hamiltonian expansions, we are able to make use of the convenient features of both. This transformation between the global and local Schrieffer-Wolff effective Hamiltonians can be treated as a quasi-local perturbation, and is thus guaranteed to preserve the gap by the topological stability theorem.
	
	Once we have demonstrated that the local SW effective Hamiltonian is in the same phase as the global SW effective Hamiltonian, we show that the higher order contributions to the local SW effective Hamiltonian can also be treated as quasi-local perturbations on the $n_*\th$ order Hamiltonian and so will not induce a phase transition. A caveat to this statement that will be made clear in the analysis is that we use a slightly modified effective Hamiltonian as compared to the standard definition of the local SW expansion. The composition of each of these results defines an adiabatic path from the $n_*\th$ order global SW effective Hamiltonian to our effective Hamiltonian at arbitrary finite order, proving \Tref{t:stability}.
	
	In addition, one would ideally like to show stability against contributions from the excited space of the unperturbed Hamiltonian, which are neglected in the effective Hamiltonian. While we expect that it may be possible to prove this kind of rigorous result using similar tools to those used here, the bounds from Ref.~\cite{Bravyi2011} on the size of these additional high-energy terms are insufficient for this purpose, and so a complete proof of stability against such terms is beyond the scope of this paper. Our analysis demonstrates that the low energy effective Hamiltonian of our model is a parent Hamiltonian for the desired state, but it does not prove that this effective Hamiltonian is a good description of the low energy physics of our system (i.e.,~that perturbation theory is accurate in this regime).  While in principle, states from the excited space of the unperturbed Hamiltonian could contribute to the low energy physics of our model, we do not know of any examples for which such unusual behaviour occurs.  We leave the investigation of such breakdowns of perturbation theory to future work.

%------------------------------------------------------------------------------------------------------------%
%------------------------------------------------------------------------------------------------------------%

\section{Perturbation Analysis}\label{s:perturb}

	\subsection{Preliminaries}\label{s:perturbintro}
	
	As in \Sref{s:model}, we define our system by a Hamiltonian of the form $H=H_0+\pert V$ with
	\begin{align}
		H_0&=\sum_{s\in \graph} Q_s\;,\qquad
		V=-\sum_e M_e \label{e:unpertHami}%\label{e:perturbdef}
	\end{align}	
	where $\pert \ll 1$. The projector to the ground space of the unperturbed Hamiltonian $H_0$ is defined as $P_0=P_\graph=\prod_sP_s$, and $P_0 H_0=H_0 P_0 =0$. It will also be convenient to define the projector to the unperturbed excited space $Q_0\equiv1-P_0$. Let $\Delta_0$ be the gap of $H_0$ and note that $\Delta_0=1$. We define $N\equiv|\graph|$ to be the total number of sites of the PEPS graph.
	
	We can motivate this choice of perturbative Hamiltonian by noticing that the low-energy behaviour of the Hamiltonian (\ref{e:fullHami}) will involve projectors to maximally entangled states of virtual qudits acting within the \physicalspace (or the unperturbed ground space), much as in the terms of the canonical parent Hamiltonian (\ref{e:injectHam}). This will allow us to argue that the low-energy effective Hamiltonian of our model is a valid parent Hamiltonian for a given PEPS, albeit in an encoded form.	We will revisit this intuitive picture before proceeding to the general analysis, but first let us define our perturbation formalism.
	
	\subsection{Global Schrieffer-Wolff perturbation expansion}\label{s:gsw}
	
		\newcommand\heff[1]{H_{\mathrm{eff}}^{\expect{#1}}}	
	
	We are interested in the low energy effective Hamiltonian of our system. To derive this effective Hamiltonian, we will make use of the global Schrieffer-Wolff (SW) perturbation expansion~\cite{Schrieffer1966,Bravyi2011}. We give a brief review of some relevant properties of the global SW method here, following Ref.~\cite{Bravyi2011}. We will focus on the relevant case where the unperturbed Hamiltonian is 1-local and the perturbation is 2-local on the lattice $\graph$, which has bounded degree. 
	
	The effective Hamiltonian derived from the global SW method is based on a transformation $\e^{S}H\e^{-S}$ that block diagonalizes $H$ with respect to the ground and excited spaces of the unperturbed Hamiltonian $H_0$. We define an anti-Hermitian operator $S$ such that
	\begin{align}\label{e:globalswrot}
		P_0\e^{S}H\e^{-S}Q_0 = Q_0\e^{S}H\e^{-S}P_0 &= 0
	\end{align}
	That is, the transformed Hamiltonian $\e^{S}H\e^{-S}$ has vanishing block-off-diagonal components.
	
	Together with (\ref{e:globalswrot}), the conditions $P_0 S P_0=Q_0 S Q_0=0$ and $\|S\|<\frac{\pi}{2}$ uniquely define $S$. We will expand $S$ in a Taylor series in $\pert$ and use this to compute an effective Hamiltonian expansion, but before we proceed we will introduce some notations. Define 
	\begin{align}
		\mathcal{L}(X) &= \frac{Q_0}{H_0}XP_0-P_0X\frac{Q_0}{H_0}\\
		X_{\mathrm{d}} &= P_0XP_0+Q_0XQ_0\\
		X_{\mathrm{od}} &=P_0XQ_0+Q_0XP_0
	\end{align}
	where we define $\frac{Q_0}{H_0}$ in the obvious way to vanish on the image of $P_0$.  
	
	Without loss of generality, we set $H_0P_0=P_0H_0=0$, i.e., the unperturbed ground state energy is set to zero.  Because of this zero eigenvalue, we note that we can express $\frac{Q_0}{H_0}$ as $Q_0\tilde{g}$ with
	\begin{align}
		\tilde{g}=\oconst P_0+\frac{Q_0}{H_0}
	\end{align}
	for an \emph{arbitrary} constant $\oconst$. We will make extensive use of this identity, and the freedom to set $\oconst$, to prove our result.
		
	 Equipped with these notations, and following Ref.~\cite{Bravyi2011}, we expand $S=\sum_{j=1}^\infty \pert^j S_j$ as a series of anti-Hermitian operators $S_j=-S_j^{\dagger}$, finding
	\begin{align}
		S_1&= \mathcal{L}(V)\label{e:sfirst}\\
		S_j&= \mathcal{L}\left([S_{j-1},V_{\mathrm{d}}]\right) + \sum_{i= 1}^{\left\lfloor\frac{j-1}{2}\right\rfloor} a_{2i}\mathcal{L}\left(W^{(j-1)}_{2i}\right) \qquad \mathrm{for}\;j>1
	\end{align}
	for $a_{i}=\frac{2^iB_i}{i!}$ with $B_i$ the Bernoulli numbers, and
	\begin{align}
		W^{(k)}_m&= \sum_{\substack{j_1,\ldots, j_{m}\geq 1\\j_1+\ldots+j_{m}=k}}[S_{j_1},[S_{j_2},[\cdots,[S_{j_{m}},V_\mathrm{od}]\cdots]]]
	\end{align}
	
	This yields an effective Hamiltonian to order $n$ of the form
	\begin{align}\label{e:globalSWdef}
		H_{\mathrm{eff}}^{\expect{n}}&= P_0 H_0 P_0 + \sum_{j=1}^n \pert^j P_0 V^{(j-1)} P_0
	\end{align}
	where
	\begin{align}
		V^{(0)} &= V \\
		V^{(j-1)} &= \sum_{i=1}^{\left\lfloor\frac{j}{2}\right\rfloor} b_{2i-1} W^{(j-1)}_{2i-1}\;,\qquad j> 1\label{e:vlast}
	\end{align}
	with $b_{2i-1}=\frac{2(2^{2i}-1)B_{2i}}{(2i)!}$.
	
	The terms in \Eref{e:globalSWdef} can be systematically calculated through a diagrammatic technique~\cite{Bravyi2011}, but for our purposes, we will not need to calculate the exact expansion of the effective Hamiltonian for a general PEPS to arbitrary order.  It will suffice for us to note that each term $P_0V^{(j)}P_0$ in (\ref{e:globalSWdef}) can be written as a linear combination of operators of the form
	\begin{align}\label{e:gammadefs}
		\Gamma({q_1,\ldots,q_{j}})&= P_0\left(\prod_{i=1}^{j}V g_{q_i}\right)V P_0
	\end{align}
	for integers $q_i$ with $\sum_{i} |q_i|=j$, and where $g_q$ is defined by
	\begin{align}
		g_0	&= P_0\\
		g_q &= \tilde{g}^q=\oconst^q P_0+\frac{Q_0}{H_0^q}&&\mathrm{for}\;q\geq 1\\
		g_q &= P_0\tilde{g}^{|q|}= \oconst^{|q|} P_0&&\mathrm{for}\;q\leq -1 \,.
	\end{align}
	 This can be seen by the application of Eqs.~(\ref{e:sfirst}-\ref{e:vlast}) and making use of the identities $\frac{Q_0}{H_0}=Q_0\tilde{g}$, $Q_0=1-P_0$, and $H_0P_0=P_0H_0=0$.	
		
	We will also make use of the fact that the effective Hamiltonian \Eref{e:globalSWdef} obeys the linked cluster theorem, which loosely states that all terms in the perturbative expansion at order $n$ are $O(n)$-local.
	
	\subsection{Ground space of the effective Hamiltonian}\label{s:gsgsw}
	
	In this section, we prove Theorem \ref{t:nthorder}. That is, we will show that to some finite order $n_*$, the effective Hamiltonian expansion of Eq.~(\ref{e:fullHami}) is a gapped parent Hamiltonian for the desired quasi-injective, isometric PEPS with a gapped canonical parent Hamiltonian. Before we begin the proof in earnest, let us briefly attempt to give some intuition for our construction.
	
	The effective Hamiltonian (\ref{e:globalSWdef}) can be written as a linear combination of $\Gamma({q_1,\ldots,q_{j}})$ operators as defined in Eq.~(\ref{e:gammadefs}). Imagine for the moment that we were able to neglect the $g_q$ terms in these operators (i.e.~neglect the dependence on the spectrum of $H_0$). These operators would then reduce to $P_0 V^{j+1} P_0$, and for isometric, quasi-injective PEPS the effective Hamiltonian would take the form
	\begin{align}
		H_{\mathrm{eff}}^{\expect{n}} &\sim -\sum_{j=0}^n \pert^n \pr_0\left(\sum_e M_e\right)^{j}\pr_0\qquad\qquad\mbox{(neglecting constants and }g_q\mbox{ factors)}\\
			&\sim-\sum_R \pmapbrac{\graph}^{\dagger} \canproj_R\pmap_\graph\label{e:effHamsimple}
		%&= H_{\mathrm{can}',v} 
	\end{align}
	where $R$ runs over all regions containing at most $n$ edges. 

	Equation (\ref{e:effHamsimple}) is precisely the encoded parent Hamiltonian (\ref{e:canparentHamv}), and so for $n\sim O(r_*)$ will have the desired PEPS as its ground state. (Recall from \Sref{s:peps} that $r_*$ is the maximum required region size to guarantee the canonical parent Hamiltonian has the correct ground space.) The following sections will be devoted to giving this simple intuition a level of rigor.	

	\subsubsection{Analysis}\label{s:perturbanal}	
		
	In order to analyse the ground space of the effective Hamiltonian (\ref{e:globalSWdef}), it will be useful to split it into two parts: $H_{\mathrm{eff}}^{\expect{n}}=\tilde{H}_{\mathrm{eff}}^{\expect{n}}+\tilde{H}_{\mathrm{else}}$, where  $\tilde{H}_{\mathrm{eff}}^{\expect{n}}$ contains all $\Gamma({q_1,\ldots,q_{j}})$ terms with all $q_i$ positive, and $\tilde{H}_{\mathrm{else}}$ contains the terms with at least one $q_i \leq 0$. The motivation for this split is that $g_q$ for $q\leq 0$ are proportional to $P_0$, while those with $q>0$ are not. This distinction will prove crucial in the analysis.
	
	 The constraint $\sum_i|q_i|=j$ implies that the only $\Gamma$ terms with all positive $q_i$ have all $q_i=1$. It is straightforward to demonstrate that $\tilde{H}_{\mathrm{eff}}^{\expect{n}}$ can be written as
	\begin{align}\label{e:restrictham}
		\tilde{H}_{\mathrm{eff}}^{\expect{n}}&= P_0 H_0 P_0 + \sum_{j=1}^n (-1)^{j}\pert^j P_0 \Gamma({1,1,\ldots 1,1}) P_0 \,,
	\end{align}
	where we have made use of the fact that $b_1=\frac{1}{2}$. The behavior of Eq.~(\ref{e:restrictham}) actually depends on the value of $\oconst$, in contrast to the complete effective Hamiltonian (\ref{e:globalSWdef}) which is independent of $\oconst$. This is because we have neglected some terms which would otherwise cancel out the effect of $\oconst$ in the Hamiltonian.
	
	Our proof of Theorem \ref{t:nthorder} will proceed in two parts. In the first, we will expand the restricted Hamiltonian $\tilde{H}_{\mathrm{eff}}^{\expect{n}}$ of Eq.~(\ref{e:restrictham}) to some finite order $n_*\sim r_*$, and show that this is a valid parent Hamiltonian for a given (quasi-injective, isometric) PEPS for sufficiently large $\oconst$. This may seem suspicious at first glance, as our proof only holds for sufficiently large values of an unphysical parameter. However, taking $r_*\sim O(1)$, $\oconst$ should be understood as a placeholder for some $O(1)$ (i.e.~intensive) constant that will be important in the subsequent analysis. Although the value of $\oconst$ chosen does not affect the behavior of the effective Hamiltonian (\ref{e:globalSWdef}), the value of $\oconst$ required to demonstrate this result captures the magnitude of a relevant energy scale in the problem.
	
	In the second part of the proof, we will restore the neglected terms $\tilde{H}_\mathrm{else}$ to analyse the complete effective Hamiltonian (\ref{e:globalSWdef}) at order $n_*$. We will show that there exists sufficiently small $\pert$ that the physics of (\ref{e:globalSWdef}) is dominated by $\tilde{H}_{\mathrm{eff}}^{\expect{n}}$. That is, the additional contributions from $\tilde{H}_{\mathrm{else}}$ do not affect the ground space nor gappedness of the Hamiltonian. The required value of $\pert$ will be set in part by the value of $\oconst$. The neglected terms that we restore in this part of the analysis have some properties that will be quite useful. Recall that every neglected term has at least one $q_i<1$, and $g_q\propto P_0$ for all $q<1$. This allows us to decompose any such term into a product of $\Gamma(1,1,\ldots,1,1)$ terms. Putting our Hamiltonian into this form (i.e.~a sum of products of $\Gamma(1,1,\ldots,1,1)$ terms) will allow us to analyse it effectively.
	
	We now prove Theorem \ref{t:nthorder}, beginning with the following lemma.
			
	\begin{lemma}\label{l:oconstlemma}
		There exists $O(1)$ (i.e.~intensive) constants $\oconst$ and $n=n_*$ such that $\tilde{H}_{\mathrm{eff}}^{\expect{n}}$ is a valid parent Hamiltonian for a given quasi-injective, isometric PEPS with ground state energy $<0$.
	\end{lemma}
	\begin{proof}	
	Consider one of the terms in the restricted effective Hamiltonian of Eq.~(\ref{e:restrictham}):
	\begin{align}
		%P_0\tilde{V}^{(j)}P_0&=(-1)^{j-1}\Gamma({1,1,\ldots 1,1})\\
		(-1)^{j}\Gamma({1,1,\ldots 1,1})&=(-1)^{j}P_0\left(\prod_{i=1}^{j-1}V \tilde{g}\right)V P_0\\
		&=-P_0\left(\prod_{i=1}^{j-1}\left(\sum_e M_e\right)\tilde{g}\right)\left(\sum_e M_e\right) P_0		
	\end{align}
	Similarly to Eq.~(\ref{e:effHamsimple}), if we could ignore the $\tilde{g}$ terms, this would become a sum of $\pmapbrac{\graph}^{\dagger} \canproj_R\pmap_\graph$ operators. Thus we begin our analysis by expanding the operator $\tilde{g}$ as follows:
	\begin{align}\label{e:dissecth0}
		\tilde{g}&=\oconst P_0+\frac{Q_0}{H_0}= \oconst{P}_0(\graph)+\frac{1}{\Delta_0}{P}_1(\graph)+\frac{1}{2\Delta_0}{P}_2(\graph)+\ldots+\frac{1}{N\Delta_0}{P}_{N}(\graph)
	\end{align}
	for ${P}_i(\graph)$ the projector to the $i\th$ excited space of $H_0$. This follows from the equally spaced spectrum of the unperturbed Hamiltonian $H_0$. In fact, we will be interested only in the effect of $\tilde{g}$ on local regions $R$, and so we will generally need only consider excited states up to energy $|R|\Delta_0$. Define the restricted operator
	\begin{align}\label{e:dissecth0restrict}
		\tilde{g}(R)&=\oconst{P}_0(R)+\frac{1}{\Delta_0}{P}_1(R)+\frac{1}{2\Delta_0}{P}_2(R)+\ldots+\frac{1}{|R|\Delta_0}{P}_{|R|}(R)
	\end{align} 	
	where the ${P}_i(R)$ involve only states with excitations localized within $R$.	
	
	Since $H_0$ is a sum of commuting projectors $Q_s$, we can easily enumerate all possible states with $i$ excitations localized within a region $R$. The corresponding projectors ${P}_i(R)$ can each be expanded as 
	\begin{align}
		%\bar{P}_0&=P_0\\
		%\bar{P}_1&=\sum_{{s_1}\in \graph}\left(Q_{s_1}\prod_{s'\in \graph\setminus {s_1}}P_{s'}\right)\\
		%&=\sum_{{s_1}\in \graph}\left(\prod_{s'\in \graph\setminus {s_1}}P_{s'}\right)-N'P_0\\
		%\bar{P}_2&=\sum_{s_1, s_2\neq s_1\in \graph}\left(Q_{s_1}Q_{s_2}\prod_{s'\in \graph\setminus s_1,s_2}P_{s'}\right)\\
		%&=\sum_{s_1, s_2\neq s_1\in \graph}\left(\prod_{s'\in \graph\setminus \{s_1, s_2\}}P_{s'}\right)-(N'-1)\sum_{{s_1}\in \graph}\left(\prod_{s'\in \graph\setminus \{s_1, s_2\}}P_{s'}\right)+\frac{1}{2}N'(N'-1)P_0\\
		{P}_i(R)&=P_{\graph\setminus R}\sum_{R'\subseteq R:|R'|=i}\left(\prod_{s'\in R'}Q_{s'}\right)\left(\prod_{s\in R\setminus R'}P_{s}\right)\label{e:38}\\
		%&=\sum_{\{s_j\}_i}\left(\prod_{s'\in \graph\setminus \{s_j\}_iP_{s'}}\right)
		%-\sum_{\{s_j\}_{i-1}}\left(\prod_{s'\in \graph\setminus \{s_j\}_iP_{s'}}\right)
		%+\sum_{\{s_j\}_{i-2}}\left(\prod_{s'\in \graph\setminus \{s_j\}_iP_{s'}}\right)
		%-\ldots +(-1)^{} \sum_{s_1}\left(\prod_{s'\in \graph\setminus s_1 P_{s'}}\right)
		%+(-1)^{} P_0
		&=P_{\graph\setminus R}\sum_{k=0}^i\left(\sum_{R''\subseteq R:|R''|=(i-k)}(-1)^k\begin{pmatrix}|R|-(i-k)\\k\end{pmatrix}\left(\prod_{s\in R\setminus R''}P_{s}\right)\right)
	\end{align}
	where $R_1\setminus R_2$ contains the sites of $R_1$ that are not in $R_2$, and we have expanded the $Q_s=1-P_s$ on the second line. Noting that $\sum_{R''\subseteq R:|R''|=j}\prod_{s\in R\setminus R''}P_{s}=\sum_{R'\subseteq R:|R'|=|R|-j}P_{R'}$, we find
	\begin{align}
		\tilde{g}(R)&= \left(\oconst+\sum_{i=1}^{|R|}(-1)^{i}\frac{1}{i\Delta_0}\begin{pmatrix}|R|\\i\end{pmatrix}\right)\cdot P_0
		\nonumber\\
		&\qquad +\sum_{j=1}^{|R|}\left(\sum_{k=0}^{|R|-j}(-1)^k\frac{1}{(j+k)\Delta_0}\begin{pmatrix}|R|-j\\k\end{pmatrix}\right)\cdot\sum_{R'\subseteq R:|R'|=|R|-j}P_{R'}P_{\graph\setminus R}\label{e:40}
	\end{align}
	
	The sums over $k$ can be evaluated explicitly, and yields the result
	\begin{align}\label{e:h0expandexplicit}
		\tilde{g}(R)&= \left(\oconst-\frac{1}{\Delta_0}\hcurl_{|R|}\right) P_0+\sum_{j=1}^{|R|}\left(j\begin{pmatrix}|R|\\j\end{pmatrix}\right)^{-1}\cdot\left(\sum_{R'\subseteq R:|R'|=|R|-j}P_{R'}P_{\graph\setminus R}\right)
	\end{align}
	for $\hcurl_{|R|}=\sum_{j=1}^{|R|}\frac{1}{j}$. We can guarantee that the first term in this expression is positive by choosing $\oconst>\frac{1}{\Delta_0}\hcurl_{|R|}$, while it is clear that $\left(j\begin{pmatrix}|R|\\j\end{pmatrix}\right)^{-1}$ is positive for all $j>0$. It is clear that as long as $|R|$ is finite, we can also choose $\oconst$ finite.
	
	Now returning to $\Gamma({1,1,\ldots 1,1})$, we notice that
	\begin{align}
		(-1)^{j}\Gamma({1,1,\ldots 1,1})&=-P_0\left(\prod_{i=1}^{j-1}\left(\sum_e M_e\right) \tilde{g}\right)\left(\sum_e M_e\right) P_0	\\
		%&= -P_0 \left(\prod_{i=1}^{j-1}\left(\sum_e M_e\right)\left(\sum_{k=0}^{m}c_k\left(\sum_{R'\subseteq R:|R'|=|R|-k}P_{R'}P_{\graph\setminus R}\right)\right) \right)\left(\sum_e M_e\right)P_0\\
		&=-\sum_{\substack{\{R_{E_i}\}\\\sum_i|R_{E_i}|=j}}\sum_{\{R_{P_i}\}} c_{\{R_{E_i},R_{P_i}\}} P_0\concatproj_{\{R_{E_i},R_{P_i}\}}P_0
	\end{align}	
	for some constants $c$, recalling the definitions of $\concatproj_{\{R_{E_i},R_{P_i}\}}$ presented in (\ref{e:defupsilons}) as part of our definition of quasi-injectivity, and using $\tilde{g}(R)$ for $\tilde{g}$ as appropriate. We have also used the fact that the global SW effective Hamiltonian obeys the linked cluster theorem so that all terms in $\tilde{H}_{\mathrm{eff}}^{\expect{n}}$ act within regions of size $O(n)$. Further, we can see by comparison with Eq.~(\ref{e:h0expandexplicit}) that for sufficiently large $\oconst$, the constants $c$ will all be positive.
	
	Thus, the restricted effective Hamiltonian at $n$th order can be expressed as
	\begin{align}\label{e:restrictHamfinish}
	\tilde{H}_{\mathrm{eff}}^{\expect{n}}&= -\sum_{j=0}^{n}\pert^j \sum_{\substack{\{R_{E_i}\}\\\sum_i| R_{E_i}|=j}}\sum_{\{R_{P_i}\}} c_{\{R_{E_i},R_{P_i}\}} P_0\concatproj_{\{R_{E_i},R_{P_i}\}}P_0
	\end{align}
	
	This Hamiltonian takes the form of \Eref{e:SHamil}, and so as demonstrated in \Sref{s:quasiinj} it will be a valid parent Hamiltonian for a quasi-injective isometric PEPS for sufficiently large $n$. We can always find some finite order $n_*\sim O(r_*)$ that will contain all regions in the canonical parent Hamiltonian (\ref{e:canparentHamp}).
	
	Additionally, because all the terms in (\ref{e:restrictHamfinish}) are negative semi-definite, the ground space energy of $\tilde{H}_{\mathrm{eff}}^{\expect{n}}$ cannot be higher than that of $\tilde{H}_{\mathrm{eff}}^{\expect{n-1}}$, and so cannot be higher than $0$.\end{proof}

	The proof of Lemma \ref{l:oconstlemma} is the main part of our analysis that depends explicitly on the spectrum of $H_0$. If one were interested in analysing a modified construction with an alternative unperturbed Hamiltonian (\ref{e:unpertHami}) that is 2-body on the virtual qudits as well as the \physicalqudits (as discussed in \Sref{s:reducecomplex}), then the analogue of Lemma \ref{l:oconstlemma} may need alternative proof techniques.
	
	We have thus far considered the restricted Hamiltonian (\ref{e:restrictham}). Now we will restore terms from $\tilde{H}_{\mathrm{else}}$ to analyse the complete effective Hamiltonian (\ref{e:globalSWdef}). We will demonstrate that there exists sufficiently small but non-zero $\pert$ such that the ground spaces of these two Hamiltonians coincide.
	
	The terms neglected in $\tilde{H}_{\mathrm{eff}}^{\expect{n}}$ are linear combinations of the terms $\Gamma(q_1,\ldots,q_j)$ where there exists some $i$ such that $q_i<1$. Since $g_{q}\propto\left(g_q\right)^2$ for $q<1$, these $\Gamma(q_1,\ldots,q_j)$ can be rewritten as $\Gamma(q_1,\ldots,q_{i-1})\Gamma(q_{i+1},\ldots,q_{j})$ up to constant factors. By making use of the decomposition (\ref{e:h0expandexplicit}) we can expand any $\Gamma(q_1,\ldots,q_j)$ into linear combinations of $\concatproj$ terms. Decomposing each $\Gamma$ in this way, we can rearrange the effective Hamiltonian into sums of terms acting on each region of the lattice. The linked-cluster theorem (together with the guarantee that the degree of $\graph$ is bounded) guarantees that for finite $n$ only a finite number of terms act on each region. Thus there exist finite constants $c_i>0$ and $\tilde{c}_i\in \R$ such that 
	\begin{align}\label{e:expandedlocsweffham}
		\heff{n}&= -\sum_{j=0}^{n}\sum_{k=0}^{n-j} \sum_{\substack{\{R_{E_i}\}\\\sum_i| R_{E_i}|=j}}\sum_{\{R_{P_i}\}} \sum_{\substack{\{R_{E_i'}\}\\\sum_i| R_{E_i'}|=k}}\sum_{\{R_{P_i'}\}}\left(\pert^j c_{\{R_{E_i},R_{P_i}\}} P_0\concatproj_{\{R_{E_i},R_{P_i}\}}P_0\right. \nonumber \\
		&\qquad\qquad\left. +\pert^k c_{\{R_{E_i'},R_{P_i'}\}} P_0\concatproj_{\{R_{E_i'},R_{P_i'}\}}P_0  +\pert^{j+k} \tilde{c}_{\{R_{E_i},R_{P_i}\},\{R_{E_i'},R_{P_i'}\}} P_0\concatproj_{\{R_{E_i},R_{P_i}\}}\concatproj_{\{R_{E_i'},R_{P_i'}\}}P_0\right)
	\end{align}
	where $\tilde{c}_{\{R_{E_i},R_{P_i}\},\{R_{E_i'},R_{P_i'}\}}=0$ if $|R_{E_i}|=0$ or $|R_{E_i'}|=0$. If we could also guarantee that all $\tilde{c}_i>0$, then the proof that $\heff{n}$ is a parent Hamiltonian would be immediate. Unfortunately this will not be the case in general, but noting that $\pert^j, \pert^k \geq \frac{1}{\pert}\pert^{j+k}$ for terms with non-vanishing $\tilde{c}$, we use following simple lemma (given without proof) to show that the ground spaces of $\heff{n}$ and $\tilde{H}_{\mathrm{eff}}^{\expect{n}}$ coincide for sufficiently small $\pert$.
		
	\begin{lemma}\label{l:gsgen}
		Consider Hermitian operators $A, B$ with eigenvalues $\eta_1(A) \geq \eta_2(A) \geq \ldots \geq \eta_n(A)$ and $\eta_1(B) \geq \eta_2(B) \geq \ldots \geq \eta_n(B)$. If $A$ and $B$ share a common eigenspace $\mathcal{H}$ with eigenvalues $\eta_1(A)$ and $\eta_1(B)$ respectively, then $\mathcal{H}$ is also an eigenspace of
		\begin{align}
		C=A-\lambda B
		\end{align} 
		corresponding to the largest eigenvalue of $C$, for $\lambda<\lambda_c = \frac{\Delta_A}{2||B||}$, with $\Delta_A = \eta_1(A)-\eta_2(A)$. Furthermore, $C$ has a finite gap $\Delta_C>2||B||(\lambda_c-\lambda)$ between largest and second largest eigenvectors.
	\end{lemma}

	Since $\pert^j, \pert^k \geq \frac{1}{\pert}\pert^{j+k}$, we can immediately apply Lemma \ref{l:gsgen} to conclude that there exists $\pert>0$ such that each bracketed term has the same ground space as if all $\tilde{c}$ were zero (since we are guaranteed that an isometric, quasi-injective PEPS state corresponds to the highest eigenvalue of any $\concatproj$ operator). Immediately we can conclude that $\heff{n}$ has the same ground space as $\tilde{H}_{\mathrm{eff}}^{\expect{n}}$. If we consider each bracketed term as a (local) operator, this Hamiltonian is also frustration-free. Since the $c$ and $\tilde{c}$ were implicitly functions of $\oconst$, the critical value of $\pert$ will similarly be set in part by the value of $\oconst$.
	
	Finally, to complete the proof of Theorem \ref{t:nthorder}, we note that because all of the $c, \tilde{c}$ coefficients in Eq.~(\ref{e:expandedlocsweffham}) are $O(1)$, the gap of the effective Hamiltonian for our system must be at least $O(\pert^{n})$. Additionally, we can similarly see that $\pert$ can be chosen small enough that $\heff{n}$ has ground space energy no greater than $\heff{0}$.  Because the energy of any state in the image of $Q_0$ is $0$, and the ground state energy of $\heff{0}$ is $0$, at least one ground state of $\heff{n}$ must be in the support of $P_0$. This means that it is sensible to discuss the restriction of Eq.~(\ref{e:expandedlocsweffham}) to $P_0$. In particular, this is useful because (\ref{e:expandedlocsweffham}) in its general form is not a sum of local terms (note that $P_0$ is a highly non-local operator, having support on the entire lattice), while it does have this feature after being restricted to $P_0$.
	
	%This can be seen from the fact that we just argued that all the negative contributions to the effective Hamiltonian can dominate, and the ground state energy of $\heff{0}$ is zero.

%------------------------------------------------------------------------------------------------------------%
%------------------------------------------------------------------------------------------------------------%
\section{Stability of Effective Hamiltonian} \label{s:stable}

	\def\l{\mathrm{loc}}
	\def\g{\mathrm{glob}}
	\newcommand\heffg[1]{H_{\mathrm{eff},\g}^{\expect{#1}}}	
	\newcommand\heffgR[1]{H_{\mathrm{eff},\g,R}^{\expect{#1}}}	
	\newcommand\heffgRs[1]{H_{\mathrm{eff},\g,R_s}^{\expect{#1}}}	
	\newcommand\heffl[1]{H_{\mathrm{eff},\l}^{\expect{#1}}}
	\newcommand\hefflR[1]{H_{\mathrm{eff},\l,R}^{\expect{#1}}}
	\newcommand\hefflRs[1]{H_{\mathrm{eff},\l,R_s}^{\expect{#1}}}
	\newcommand\hl[1]{H_{\l}^{\expect{#1}}}

	Having proved \Tref{t:nthorder}, we will now show that the ground space of $\heff{n}$ does not change dramatically (i.e.,~will remain in the same phase) as we include additional contributions from higher order terms in the perturbation expansion, as detailed in Theorem~\ref{t:stability}.  In this context, we will say that $\heff{n}$ is \emph{stable} to these additional terms. To prove Theorem~\ref{t:stability}, we will require that the canonical parent Hamiltonian of the PEPS under consideration obeys the Local-TQO and Local-Gap conditions, and so the PEPS is topologically ordered. This will allow us to use the results on stability of topologically ordered systems under quasi-local perturbations~\cite{Michalakis2011,Bravyi2010}.
	
	Let us outline the proof strategy, which proceeds in two stages. In order to analyse the stability of the effective Hamiltonian, it will be convenient to make use of the \emph{local} Schreiffer-Wolff perturbation method in contrast to the global SW method used in \Sref{s:perturb}. The local SW method produces an effective Hamiltonian expansion with locality properties that we will exploit to prove our stability results. With this in mind, the first stage of our proof will be to transform the global SW effective Hamiltonian derived in the previous section into the corresponding local SW effective Hamiltonian and showing that the properties of the ground space are preserved under this transformation. This is captured by the following lemma.
	\begin{lemma}\label{l:globaltolocal}
		For a quasi-injective, isometric, topologically ordered PEPS, the global SW effective Hamiltonian of our model at order $n_*$ is in the same phase as the local SW effective Hamiltonian at order $n_*$ for sufficiently small $\pert>0$.
	\end{lemma}
	
	In order to prove this lemma, we will show that the transformation between the global and local effective Hamiltonians can be achieved by the addition of a sufficiently small quasi-local operator. This allows us to use the topological stability theorem to argue that the two Hamiltonians are in the same phase. We then also give a lemma showing that if the global SW effective Hamiltonian is topologically stable, then so is the local SW effective Hamiltonian.
	
	At this point in the analysis we will simply have demonstrated that the ground space of one finite order effective Hamiltonian is in the same phase as the ground space of another finite order effective Hamiltonian. For the second stage of our proof, we will use the structure of the local SW perturbation expansion to argue that the higher-order contributions to the local SW Hamiltonian are both small and quasi-local. This will allow us to again apply the topological stability theorem and prove the following:
	\begin{lemma}\label{l:highorderstable}
		For a quasi-injective, isometric, topologically ordered PEPS, the local SW Hamiltonian of our model at order $n_*$ is in the same phase as our effective Hamiltonian at any order $k\geq n_*$, for sufficiently small $\pert>0$.
	\end{lemma}

	The two lemmas \ref{l:globaltolocal}-\ref{l:highorderstable} constitute a proof of \Tref{t:stability}.
	
	Throughout the following analysis, we will make use of the fact that both $n_*$ and the maximum coordination number of $\graph$ are $O(1)$ constants in $N$ and $\pert$. We will also often use locality properties of operators in this section. Although the majority of the Hamiltonians and operators we consider in this section are highly non-local (e.g.,~the unperturbed ground space projector $P_0$), we will often use the fact that these operators are local when restricted to the image of $P_0$. We will often loosely refer to an operator as local or quasi-local, when in fact it is clear from context that this is only true after the restriction to the unperturbed ground space.	

	\subsection{Transformation to local Schrieffer-Wolff effective Hamiltonian}
	
		To analyse the stability properties of the effective Hamiltonian (\ref{e:globalSWdef}), we will consider a related effective Hamiltonian derived from the local Schrieffer-Wolff method~\cite{Datta1996, Bravyi2011}.  (The previous method has been referred to as the global SW method to avoid confusion.)  As with the global SW method, the local SW method is based on a transformation that block diagonalizes the Hamiltonian with respect to the ground space and excited space of the unperturbed Hamiltonian. In contrast to the global SW transformation, the local SW transformation does not achieve this block diagonalization exactly, but only up to corrections of order $O(\pert^{n+1})$ for a given order $n$. However, it is constructed in a manifestly local way, which allows us to analyse some properties of this expansion much more directly.
	
	\subsubsection{Local Schrieffer-Wolff transformation}	
	Before proceeding, we will briefly define and review some relevant properties of the local SW method, following Ref.~\cite{Bravyi2011}. At a given order $n$ we construct a sequence of anti-Hermitian operators
	\begin{align}
	T^{\expect{n}}&=\sum_{q=1}^n \pert^q T_q
	\end{align}
	such that all $T_q$ are $(q+1)$-local and
	\begin{align}\label{e:smalloffdiag}
		\left\|P_0\e^{T^{\expect{n}}}H\e^{-T^{\expect{n}}}Q_0 + Q_0\e^{T^{\expect{n}}}H\e^{-T^{\expect{n}}}P_0\right\| &\leq O(N\pert^{n+1})
	\end{align}
	for sufficiently small $\pert$. We can decompose the transformed Hamiltonian into a block-diagonalized part and a ``garbage'' part as
	\begin{align}\label{e:rotatedfulllocalSW}
		\e^{T^{\expect{n}}}H\e^{-T^{\expect{n}}}&=\hl{n}+H_{\mathrm{garbage}}
	\end{align}
	where $Q_0\hl{n}P_0=P_0\hl{n}Q_0=0$. We use the subscript `$\l$' to denote operators derived from the local SW method where there may be confusion with similar operators from the global SW method. Because $\hl{n}$ is block-diagonal, \Eref{e:smalloffdiag} implies that $\|H_{\mathrm{garbage}}\|$ is $O(N\pert^{n+1})$.
	
	The effective Hamiltonian at order $n$ is defined as the restriction of $\hl{n}$ to the ground space of the unperturbed Hamiltonian $H_0$:
	\begin{align}
		\heffl{n} &=P_0 \hl{n} P_0
	\end{align}
	
	In order to explicitly compute $\hl{n}$, we first define a series of Hermitian operators
	\begin{align}
		V^{(j)}_{\l}&=\sum_{q=2}^{j+1}\frac{1}{q!}\sum_{\substack{1\leq j_1,\ldots,j_q\leq n \\j_1+\ldots+j_q=j+1}}[T_{j_1},[T_{j_2},[\cdots,[T_{j_q},H_0]\cdots]]]
		+\sum_{q=1}^{j}\frac{1}{q!}\sum_{\substack{1\leq j_1,\ldots,j_q\leq n \\j_1+\ldots+j_q=j}}[T_{j_1},[T_{j_2},[\cdots,[T_{j_q},V]\cdots]]]\label{e:vjdef}
	\end{align}
	where $V^{(0)}_{\l}=V$. We represent each of these operators as a sum of local terms
	\begin{align}
		V^{(j)}_{\l}=\sum_{R\subseteq \graph}V^{(j)}_{R,\l}\label{e:defVdecomp}
	\end{align}
	where each $V^{(j)}_{R,\l}$ is Hermitian and acts non-trivially only on spins within region $R$. This decomposition is chosen as the expansion of $V^{(j)}_{\l}$ in some orthogonal operator basis. Is can be shown~\cite{Bravyi2011} that $V^{(j)}_{\l}$ is $(j+2)$-local, so we can guarantee that this local decomposition need only consider $(j+2)$-local regions $R$. Each of the $T_j$ operators take the form
	\begin{align}
		T_j&= \sum_{R\subseteq \graph}\left(\frac{Q_R}{H_R}V^{(j-1)}_{R,\l}P_R-P_RV^{(j-1)}_{R,\l}\frac{Q_R}{H_R}\right)\label{e:tjdef}
	\end{align}
where $P_R\equiv\prod_{s\in R}P_s$, $Q_R\equiv 1-P_R$, and $H_R\equiv \sum_{s\in R}Q_s$. Equations (\ref{e:vjdef}) and (\ref{e:tjdef}) can be solved recursively. Given this solution, we find
	\begin{align}\label{e:definehln}
		\hl{n}&= H_0 +\sum_{j=1}^n\pert^j\sum_{R\subseteq \graph} \left( P_RV_{R,\l}^{(j-1)}P_R+Q_RV_{R,\l}^{(j-1)}Q_R \right)
	\end{align}
%, and define $\frac{1}{H_R}$ to act identically on ground states of $H_R$ (in analogy to $\frac{1}{H_0}$)	

	\subsubsection{Properties of the local Schrieffer-Wolff transformation}

	Here we will state a number of known properties of the local Schrieffer-Wolff transformation that will be useful in our analysis~\cite{Bravyi2011}. 
	
	Although Eqns.~(\ref{e:vjdef}) and (\ref{e:tjdef}) define $V^{(j)}_{\l}$ for arbitrary (positive, integral) $j$, only those $V^{(j)}_{\l}$ with $j< n$ appear in $\hl{n}$. The remaining terms can be used to write $H_{\mathrm{garbage}}$ as
	\begin{align}\label{e:garbagedef}
		H_{\mathrm{garbage}}&= \sum_{j=n+1}^{\infty}\pert^j V^{(j-1)}_{\l}
	\end{align}
	As would be expected, $V^{(j)}_{\l}$ is independent of $n$ for $j<n$, while it is implicitly dependent on $n$ for $j\geq n+1$. For sufficiently small $\pert$, the norm of $H_{\mathrm{garbage}}$ can be bounded as $\|H_{\mathrm{garbage}}\|\leq c\Delta_0 N\pert^{n+1}$ for a constant $c$ that depends on $n$.
	
	The local SW method obeys the linked-cluster theorem, and so we can guarantee that $\heffl{n}$ is $O(n)$-local. It is also the case that at a fixed order $n$, the effective Hamiltonians found by the global and local SW methods can be related by a transformation $K^{\expect{n}}$ up to an error
	\begin{align}
	    \label{eq:delta}
			\heffg{n}-\e^{K^{\expect{n}}}\heffl{n}\e^{-K^{\expect{n}}}\equiv \hat{\delta}
	\end{align}
	with $\|\hat{\delta}\|\leq O(N|\pert|^{n+1})$ for a system with $N$ sites and $|\pert|<1$, and where $K^{\expect{n}}$ is $O(n)$-local. This is shown in Ref.~\cite[Lemma 4.4]{Bravyi2011}. We denote the global SW effective Hamiltonian by $\heffg{n}$ and the local SW effective Hamiltonian by $\heffl{n}$ to avoid confusion. In the following analysis, we will prove that $\hat{\delta}$ is quasi-local with favourable decay parameters. In doing so, we will prove Lemma \ref{l:globaltolocal}.
	
	Because $\heffg{n}$, $\heffl{n}$ and $K^{\expect{n}}$ are all $O(n)$-local, we can decompose them into operators acting non-trivially only on connected regions of $\graph$ with bounded size. Denote such a decomposition of an operator $X$ as $X=\sum_R X_R$. We now define a bound on the strength of an operator as
	\begin{align}
		\|X\|_{\mathrm{max}}&= \max_{s\in \graph}\left\|\sum_{R\ni s}X_R\right\|
	\end{align}
	
	It can be shown that $\|K^{\expect{n}}\|_{\mathrm{max}}=O(|\pert|)$~\cite{Bravyi2011}. In fact, we can expand $K^{\expect{n}}$ as a Taylor series in $\pert$ as
	\begin{align}
		K^{\expect{n}} &=\sum_{j=1}^n\pert^j K_j
	\end{align}
	for some $O(j)$-local, block-diagonal $K_j$ with $\|K_j\|_{\mathrm{max}}\sim O(1)$.
	
	\subsubsection{Transforming from global to local Schrieffer-Wolff effective Hamiltonians}	

	Given the properties of the local Schrieffer-Wolff transformation noted above, we now demonstrate that $\hat{\delta}$ of Eq.~(\ref{eq:delta}), the operator that relates the local and global SW effective Hamiltonians, is quasi-local. This quasi-locality will allow us to argue that $\heffg{n}$ is stable under addition of $\hat{\delta}$, and thus the gap does not close along a path from $\heffg{n}$ to $\heffl{n}$.
	
	\begin{lemma}\label{l:deltaquasilocal}
		$\hat{\delta}$ is $\left(O(1),O(\pert^{O(1)})\right)$-quasi-local when restricted to the space $P_0$.
	\end{lemma}
	\begin{proof}
	Recall from our definition of quasi-locality in Sec.~\ref{s:peps}, a $(J,\mu)$-quasi-local operator has interaction strength that decays with radius $r$ as $J\mu^r$. In order to show that $\hat{\delta}$ is quasi-local, we will explicitly construct a local decomposition for it. For this purpose it will be convenient to introduce operators
	\begin{align}
		\Theta_R(k) = \sum_{j=0}^k\frac{1}{j!}[K^{\expect{n}},\cdot]^j \hefflR{n}
	\end{align}
	for $\heffl{n}=\sum_R\hefflR{n}$ an $O(n)$-local decomposition of $\heffl{n}$ with $\|\hefflR{n}\|=O(1)$, and where $[A,\cdot]^jB$ is the $j$-fold nested commutator of $A$ and $B$, e.g.~$[A,\cdot]^0B= B$,  $[A,\cdot]^1B= [A,B]$,  $[A,\cdot]^2B= [A,[A,B]]$, etc.
	
	Note that $\Theta_R(\infty)=\e^{K^{\expect{n}}}\hefflR{n}\e^{-K^{\expect{n}}}$. Because $K^{\expect{n}}$ and $\hefflR{n}$ are both $O(n)$-local (when restricted to the image of $P_0$), this leads us to consider $\Theta_R(k)$ to be a $O(kn)$-local truncation of $\e^{K^{\expect{n}}}\hefflR{n}\e^{-K^{\expect{n}}}$. To relate this new operator to $\hat{\delta}$, it is convenient to rewrite the global SW effective Hamiltonian as~\cite{Bravyi2011}
	\begin{align}
		\heffg{n}=\sum_{j=0}^{n}\sum_{\substack{1\leq q_0, q_1,\ldots,q_j\leq n \\q_0+\ldots+q_j\leq n}}\frac{1}{j!}\pert^{q_0+\ldots+q_j}[K_{q_1},[K_{q_2},\ldots[K_{q_j}, \heffl{q_0}-\heffl{q_0-1}]\ldots]]
	\end{align}
	
	Because we are interested in the local decomposition of $\hat{\delta}$, we also define
	\begin{align}
		\heffgR{n}=\sum_{j=0}^{n}\sum_{\substack{1\leq q_0, q_1,\ldots,q_j\leq n \\q_0+\ldots+q_j\leq n}}\frac{1}{j!}\pert^{q_0+\ldots+q_j}[K_{q_1},[K_{q_2},\ldots[K_{q_j}, \hefflR{q_0}-\hefflR{q_0-1}]\ldots]]
	\end{align}	
	where we note that $\heffgR{n}$ need not act only within $R$, even when restricted to the $P_0$ subspace (though it is local). It is then straightforward to see that the difference between $\heffgR{n}$ and $\Theta_R(n)$ consists only of those terms in the sum with $\sum_i{q_i}>n$, i.e.~
	\begin{align}\label{e:diffM}
		\heffgR{n} - \Theta_R(n) = -\sum_{j=0}^n\frac{1}{j!}\sum_{\substack{1\leq q_0, q_1,\ldots,q_j\leq n \\q_0+\ldots+q_j> n}}\pert^{q_0+\ldots+q_j}[K_{q_1},[K_{q_2},\ldots[K_{q_j}, \hefflR{q_0}-\hefflR{q_0-1}]\ldots]]
	\end{align}
	
	The norm of this difference is
	\begin{align}
		\|\heffgR{n} - \Theta_R(n)\| \sim O(\pert^{n+1})
		% \pert^{n+1}\sum_{j=0}^n\frac{1}{j!}\|[K^{\expect{n}},\cdot]^j \hefflR{n}\|\\
		%&\leq \pert^{n+1}n\|\hefflR{n}\|
%		&
	\end{align}
	for sufficiently small $\pert$, since each $\hefflR{q_0}$ fails to commute with at most a constant number of local terms in each $K_{q}$.
	
	We now define a local decomposition of $\hat{\delta}$ making use of these $\Theta_R(n)$ operators. This decomposition is not unique, and we may choose it as convenient so long as it has the property that $\hat{\delta}=\sum_{s,r}\hat{\delta}_{s,r}$ for $\hat{\delta}_{s,r}$ acting only within a region of radius $r$ around site $s$. For our purposes, we are interested mainly in decay of $\hat{\delta}$ on long length scales, and so we simply collect all the terms with radius smaller than some critical length scale $\kappa$. Specifically, we choose $\kappa \sim n^2$ as the maximum radius of operators $\heffgR{n} - \Theta_R(n)$ over all $R$. With this in mind, for all $s$ we define
	\begin{align}
		\hat{\delta}_{s,r} &= 0 \qquad\mbox{for}\;r<\kappa\\
		\hat{\delta}_{s,\kappa} &= \heffgRs{n} - \Theta_{R_s}(n)
	\end{align}	
	where the regions $R_s$ here have been put into one-to-one correspondence with the sites $s$ in some canonical way.
	
	In a similar spirit, we need not define $\hat{\delta}_{s,r}$ for all $r$. Instead, we will only define it for some set of radii $r_k$ for each $k>n$, as follows
	\begin{align}
		\hat{\delta}_{s,r_k} = \Theta_{R_s}(k) - \Theta_{R_s}(k+1)\qquad\mbox{for}\;k>n
	\end{align}
	such that $\hat{\delta}_{s,r_k}$ acts within radius $r_k$ of site $s$ as required. Since $\Theta_{R}(k)$ is $O(kn)$-local, this implies that $r_k\sim O(kn)$. It can clearly be seen that $\sum_{s,r}\hat{\delta}_{s,r}=\hat{\delta}$. Given the facts that $\|\hat{\delta}_{s,\kappa}\|\sim O(\pert^{n+1})$ and (since $\|K^{\expect{n}}\|\sim O(\pert)$ and $n, \|\hefflR{n}\|=O(1)$)
	\begin{align}
		\|\hat{\delta}_{s,r_k}\|\leq \|[K^{\expect{n}},\cdot]^{k+1}\hefflRs{n}\|\leq O(\pert^{k+1})
	\end{align}
	we conclude that $\hat{\delta}$ is $\left(O(1),O(\pert^{O(1)})\right)$-quasi-local (when acting on the space $P_0$) as claimed.\end{proof}
	
	Now we can make use of the fact that $\|\hat{\delta}_{s,r}\|\leq O(\pert^{n+1})$ for all $s,r$ to provide some alternative quasi-local parameters for $\hat{\delta}$. Consider the following Lemma (presented without proof).
	
	\begin{lemma}\label{l:altdecaybound}
		Consider a function $f(r)$ where $f(r)\leq ab^r$ and $f(r)\leq c$ for $0<b<1$ and $a,c>0$. Then $f(r)\leq c^{1-\lambda}a^{\lambda} b^{r\lambda}$ for all $0<\lambda<1$.
	\end{lemma}
%	\begin{proof}
%		There exists $\tilde{r}=\log_b\frac{c}{a}$ such that $ab^r<c$ for $r>\tilde{r}$. Then for all $r>\tilde{r}$ we also have
%		\begin{align}
%			ab^r&=ab^{r\lambda}b^{(1-\lambda)\log_b\frac{c}{a}}\\
%			&\leq ab^{r\lambda}\left(\frac{c}{a}\right)^{1-\lambda}\\
%			&= c^{1-\lambda}a^{\lambda} b^{r\lambda}
%		\end{align}
%		and for $r\leq\tilde{r}$ we have
%		\begin{align}
%			c\leq c^{1-\lambda}a^{\lambda} b^{r\lambda}
%		\end{align}
%		Since $f(r)\leq ab^r$ and $f(r)\leq c$, this completes the proof.
%	\end{proof}
	
	This implies that $\hat{\delta}$ is $\left(O(\pert^{(n+1)(1-\lambda)}),O(\pert^{\lambda O(1)})\right)$-quasi-local for any choice of $0<\lambda<1$. Particularly, let us choose $\lambda<\frac{1}{n+1}$. Importantly, this means that by choosing $\pert$ sufficiently small, we can make the first parameter of the quasi-local decay arbitrarily small compared to the $O(\pert^n)$ gap of $\heffg{n}$, and the second parameter can be made arbitrarily small simultaneously by decreasing $\pert$.
	
	Because $\heffg{n_*}$ is frustration free, satisfies Local-TQO and Local-Gap by assumption, and has a gap of $O(\pert^{n_*})$, we can apply the topological stability theorem of Ref.~\cite{Michalakis2011} for sufficiently small $\pert$ to show that the gap of $\heffg{n_*}$ remains $O(\pert^{n_*})$ along a path to $\e^{K^{\expect{n_*}}}\heffl{n_*}\e^{-K^{\expect{n_*}}}$. In particular, this shows that $\e^{K^{\expect{n_*}}}\heffl{n_*}\e^{-K^{\expect{n_*}}}$ is in the same phase as $\heffg{n_*}$.
	
	To complete the proof of Lemma \ref{l:globaltolocal}, we appeal to the following lemma, as proven in Ref.~\cite{Xie2010}
	
	\begin{lemma}[\cite{Xie2010}]\label{l:unrotate}
		Given a gapped quasi-local Hamiltonian $H$ and local Hamiltonian $X$, the ground states of $\e^{iX}H\e^{-iX}$ are in the same phase as those of $H$.
	\end{lemma}
		
	Because $iK^{\expect{n}}$ can be regarded as a local Hamiltonian, the proof of Lemma \ref{l:globaltolocal} is an immediate corollary to Lemma \ref{l:unrotate}. That is, the local SW effective Hamiltonian at order $n_*$ has the desired quasi-injective, isometric, topologically ordered PEPS as its ground state. This is the main result of this section.
	
	Before we proceed, it will be useful to demonstrate an additional property of the effective Hamiltonian $\heffl{n_*}$. In the following sections, we would like to apply the topological stability theorem (\Tref{t:topologicalstability}) to demonstrate that the ground space of $\heffl{n_*}$ is stable against some additional terms. Since the Local-TQO property itself is stable against perturbations which do not close the (local) gap~\cite{Cirac2013}, we can also argue that $\heffl{n_*}$ satisfies Local-TQO and Local-Gap. Unfortunately we have not shown that $\heffl{n}$ is frustration-free, and so we cannot directly apply \Tref{t:topologicalstability}. However, the following Lemma will allow us to leverage the topological stability of $\heffg{n_*}$ to prove topological stability of $\heffl{n_*}$.
	\begin{lemma}\label{l:localstable}
		Consider a Hamiltonian $H_{TO}$ satisfying the assumptions of \Tref{t:topologicalstability}. That is, for any quasi-local perturbation $V$ there exists some $\pert<0$ such that $H_{TO}+\pert V$ is in the same phase as $H_{TO}$. Then for each Hamiltonian $H'$ in the same phase as $H_{TO}$, there exists some $\pert'<0$ such that $H'+\pert' V$ is in the same phase as $H'$.
	\end{lemma}
	\begin{proof}
		Define a smooth, invertible, linear quasi-local transformation $\mathcal{T}$ that relates $H_{TO}$ and $H'$, i.e.~$\mathcal{T}(H_{TO})=H'$. We are guaranteed that such a transformation exists from the fact that $H'$ and $H_{TO}$ are in the same phase. This implies that $\mathcal{T}^{-1}(V)\equiv V'$ is quasi-local and $\mathcal{T}^{-1}(H'+\pert' V)=H_{TO}+\pert' V'$. Because there exists some $\pert'$ such that $H_{TO}+\pert' V'$ is in the same phase as $H_{TO}$, this also implies that $H'+\pert' V$ is in the same phase as $H'$.
\end{proof}

	\subsection{Stability to higher order contributions}	
	
	Now that we have shown that $\heffl{n_*}$ is in the same phase as $\heffg{n_*}$, we can make use of the explicit locality structure of the local SW transformation to bound the effect of higher order contributions to the effective Hamiltonian. The main technical result we will derive here is the fact that $H_{\mathrm{garbage}}$ is quasi-local, and so the ground state and gap of $\heffl{n_*}$ are stable under addition of this garbage term. Following this, we can define a sequence of effective Hamiltonians to arbitrary order (very similar to the local SW effective Hamiltonians) that are in the same phase as $\heffl{n_*}$ to arbitrary finite order, completing the proof of Lemma \ref{l:highorderstable}. 	
	
	Our analysis is based on the fact that $\hl{n}$ obeys the linked cluster theorem. This follows directly from the fact that $V_{\l}^{(j-1)}$ is $(j+1)$-local. Any non-local term arising in the expansion of $V_{\l}^{(j)}$ must vanish. Defining $V^{(j)}_{s,\l}=\sum_{R\ni s} V^{(j)}_{R,\l}$, we can also bound the strength of $V_{\l}^{(j)}$ as
	\begin{align}
		\|V_{\l}^{(j)}\|_{\mathrm{max}}= \max_{s\in \graph}\|V^{(j)}_{s,\l}\| \leq \alpha\left(\frac{n^2}{\Delta_0\beta}\right)^j
	\end{align}
	for some constants $\alpha, \beta > 0$~\cite{Bravyi2011}.
	
	Now, because $V^{(j)}_{R,\l}$ are components of $V^{(j)}_{\l}$ in an orthogonal operator basis, removing some set of them from $V^{(j)}_{s,\l}$ cannot increase $\|V^{(j)}_{s,\l}\|$. That is, $\left\|\sum_{R'\subseteq \{R:R\ni s\}}V^{(j)}_{R',\l}\right\|\leq\|V^{(j)}_{s,\l}\|$.
	
	This allows us to define a decomposition of $H_{\mathrm{garbage}}$ (recall \Eref{e:garbagedef}) into terms acting within a region of radius $r$ around each site $s$. These will be operators of the form $\pert^{r-1}\sum_{R'\subseteq \{R:R\ni s\}}V^{(r-2)}_{R',\l}$. For a fixed $n$, we then have that for $\pert< \frac{\Delta_0\beta}{n^2}$ we can bound the norms of these operators by the exponential decay 
\begin{equation}
	\Bigl\|\pert^{r-1}\sum_{R'\subseteq \{R:R\ni s\}}V^{(r-2)}_{R',\l}\Bigr\| \leq \left(\frac{\alpha\Delta_0^2\beta^2}{\pert n^{4}}\right)\left(\frac{\pert n^2}{\Delta_0\beta}\right)^{r}\,.
\end{equation} 
Therefore, $H_{\mathrm{garbage}}$ is $\left(O(\pert^{-1}),O(\pert)\right)$-quasi-local.
	
	Because we also know that $\|H_{\mathrm{garbage}}\|$ does not include terms $V^{(r-2)}$ for $r\leq n+1$, we can give an alternative bound on the decay parameters of $H_{\mathrm{garbage}}$. Making use of Lemma \ref{l:altdecaybound}, we find that $H_{\mathrm{garbage}}$ is also $\left(O(\pert^{n+1-\lambda(n+2)}),O(\pert^{\lambda})\right)$-quasi-local.  Importantly, for $\lambda<\frac{1}{n+2}$, the first of these parameters is $O(\pert^{n+\delta})$ for $\delta>0$, and both parameters can be made arbitrarily small by decreasing $\pert$. This is convenient as it allows us to use the topological stability theorem to analyse the stability of $\heffl{n}$ to contributions from $H_{\mathrm{garbage}}$. 
	
	As the gap of $\heffl{n_*}$ is $O(\pert^{n_*})$, by making $\pert$ small enough, we can make the strength of $H_{\mathrm{garbage}}$ arbitrarily small compared to the gap for $\lambda<\frac{1}{n_*+2}$. By the topologically stability theorem, we can find $\pert>0$ such that $\heffl{n_*}+H_{\mathrm{garbage}}$ is in the same phase as $\heffl{n_*}$ and has a gap of $O(\pert^{n_*})$.
	
	Given this result, we can write a sequence of effective Hamiltonians of the form 
		\begin{align}
		H_{\mathrm{eff,loc+}}^{\expect{k}}&\equiv P_0e^{T^{\expect{k}}}H e^{-T^{\expect{k}}}P_0\\
		&=H_{\mathrm{eff,loc}}^{\expect{k}}+P_0H_{\mathrm{garbage}}^{\expect{k}}P_0
	\end{align}
	where we note that $H_{\mathrm{garbage}}$ has previously been implicitly dependent on the order of perturbation theory, and so we restore this explicit dependence here. Given the fact that $\heffl{n_*}+H_{\mathrm{garbage}}^{\expect{n_*}}$ is in the same phase as $\heffl{n_*}$, we can now write for any $k\geq n_*$
	\begin{align}
		H_{\mathrm{eff,loc+}}^{\expect{k}}&=P_0e^{T^{\expect{k}}}e^{-T^{\expect{n_*}}}\left(\hl{n_*}+H_{\mathrm{garbage}}^{\expect{n_*}}\right)e^{T^{\expect{n_*}}}e^{-T^{\expect{k}}}P_0\\
		&=P_0e^{T^{\expect{k}}}e^{-T^{\expect{n_*}}}\left(\heffl{n_*}+H_{\mathrm{garbage}}^{\expect{n_*}}\right)e^{T^{\expect{n_*}}}e^{-T^{\expect{k}}}P_0
	\end{align}
	Because $iT^{\expect{k}}$ are local Hamiltonians for any finite $k$, we can appeal to Lemma \ref{l:unrotate} to show that $H_{\mathrm{eff,loc+}}^{\expect{k}}$ and $\heffl{n_*}$ are indeed in the same phase when restricted to $P_0$, and because both Hamiltonians act trivially outside of $P_0$, this concludes the proof of both Lemma \ref{l:highorderstable} and \Tref{t:stability}.
	
%------------------------------------------------------------------------------------------------------------%
%------------------------------------------------------------------------------------------------------------%

\section{Discussion}\label{s:discuss}

	Our construction yields a 2-body Hamiltonian described by \Eref{e:fullHami} that is a gapped parent Hamiltonian for a state in the same phase as a desired PEPS to all orders of perturbation theory.  We have made use of the fact that two ground states of gapped local Hamiltonians are in the same phase if one Hamiltonian can be smoothly deformed into the other without closing the gap. This also implies that expectation values of local observables deform smoothly along the same path.  Because in the limit $\pert\to 0$ both the global and local Schrieffer-Wolff transformations tend to the identity, our construction gives a parent Hamiltonian for a state which tends towards the desired PEPS state as $\pert\to 0$.  Expectation values of local observables can therefore be made arbitrarily close to those of the PEPS under consideration by choosing $\pert$ arbitrarily small.

	\subsection{Locality on virtual qudits}\label{s:reducecomplex}
	
	Our model as defined in \Sref{s:model} gives a 2-local Hamiltonian where the code gadgets of our construction are considered as indivisible quantum systems. Instead of treating a code gadget as a single quantum system, we could also be interested in implementing the virtual qudits as distinct physical systems. The code gadget Hamiltonian as defined at a site $s$ would then involve $\mathrm{deg}(s)$-body interactions in general.  There are two strategies that can be applied to also reduce these interactions to 2-body terms on the virtual qudits. The first approach is more elegant but less general, while the second is universally applicable.
	
	The most important feature of a code gadget that we must preserve in a procedure like this is its ground space. With this in mind, the first strategy involves simply finding an explicit 2-body Hamiltonian whose ground space is identical to the code gadget Hamiltonian in \Eref{e:unpertHami}. This is the approach taken in Refs.~\cite{Bartlett-2006} and \cite{Brell-2011}. When making use of this strategy, the modification of the spectral structure of the code-gadget Hamiltonian means that our technical proofs in \Sref{s:perturb} are not immediately applicable (particularly Lemma \ref{l:oconstlemma}). However, it seems reasonable to conjecture that the construction is insensitive to these details and similar results could be found for any sensible choice of code gadget Hamiltonian.
	
	The second strategy is to simply apply more conventional perturbation gadget techniques to the code gadget Hamiltonian directly, reducing it to 2-body interactions using further perturbative ancillae (for example following Ref.~\cite{Jordan2008}). Because these perturbation gadgets could be applied within each code gadget separately, they would be approximating systems involving only a fixed finite number of qudits. For this reason, many of the difficulties with applying general perturbation gadgets to infinite systems could be avoided. Additionally, because the effective Hamiltonian of this system will be identical to \Eref{e:unpertHami}, we expect that the proofs in \Sref{s:perturb} could be adapted to this situation (given an appropriately chosen perturbative hierarchy). Although it could be applied to an arbitrary system, this is clearly the less elegant option.

	\subsection{Symmetries of the model}\label{s:symmetry}
	
		In previously studied examples of these techniques~\cite{Bartlett-2006, Brell-2011}, the local symmetries of the states were exactly captured by this construction. That is, for each local symmetry of these models, a corresponding encoded symmetry can be found which commutes with the full Hamiltonian of the system, including the perturbative couplings. In generic perturbative approaches one would expect only to recover these symmetries approximately.
	
	Exactly capturing the local symmetries of the target states is not a feature of our construction in general. There exist cases where some or all of the local symmetries are exactly captured, but this does not seem to be generic. An example of this is shown in \Sref{s:egds}, where a subset of the full symmetries of the model are preserved exactly, and the rest only preserved approximately.

	\subsection{Application to RVB states}\label{s:rvb}

		Constructing and analysing parent Hamiltonians for resonating valence bond (RVB) states is of interest for modelling spin liquids and other exotic quantum phases.  Methods to construct such parent Hamiltonians for the Kagome lattice require at least 12-body interactions~\cite{Seidel2009}. Recently, an alternative construction for parent Hamiltonians of these states has been proposed, based on a PEPS representation of the RVB states on this lattice~\cite{Schuch2012}. Canonical parent Hamiltonians for this PEPS require at least 19-body interactions.
		
		Because the RVB PEPS is $\Z_2$-injective~\cite{Schuch2012}, we anticipate that our analysis could be applied to this case and as such a 2-body parent Hamiltonian of the form (\ref{e:fullHami}) may be obtained for a state in the same phase as the RVB state (and in the limit that the perturbation parameter vanishes, should reproduce the RVB state precisely).  In this context, we remind the reader that our construction yields an encoded version of the desired state.  The dimension of the Hilbert spaces associated with each site will be larger, and there will also be ancilla systems required to mediate coupling between the sites (corresponding to tensors with no physical indices in the description of Ref.~\cite{Schuch2012}).

%------------------------------------------------------------------------------------------------------------%
%------------------------------------------------------------------------------------------------------------%

\section{Example: The Double Semion Model}\label{s:egds}
	
		We now present an illustrative example of our construction.	The double semion model is a simple example of a string-net model~\cite{LevinSN05}, whose ground states are known to have exact PEPS descriptions~\cite{Gu-09, Buerschaper-09}. In fact, the double semion model has a particularly simple PEPS description~\cite{Gu2008, Gu-09} that we can exploit to construct a 2-body system whose low energy effective Hamiltonian is an encoded parent Hamiltonian for the double semion ground space.	These states are examples of both $(G,\omega)$-injective PEPS and MPO-injective PEPS.
		
		We will use the double semion model to demonstrate some of the features of our construction. The analysis in this section should be understood as illustrative of the features of the Hamiltonian (\ref{e:fullHami}) rather than as an example of the theorems proven in Sections \ref{s:perturb} and \ref{s:stable}. As such, we make use of a simplified formalism that sacrifices some rigor for clarity. The general analysis as shown in Sections \ref{s:perturb} and \ref{s:stable} can be applied to this example to demonstrate the relevant features more rigorously.
				
%------------------------------------------------------------------------------------------------------------%
		\begin{figure}
		\centering
		\subfloat[]{\label{F:DShexlattice}
		\includegraphics{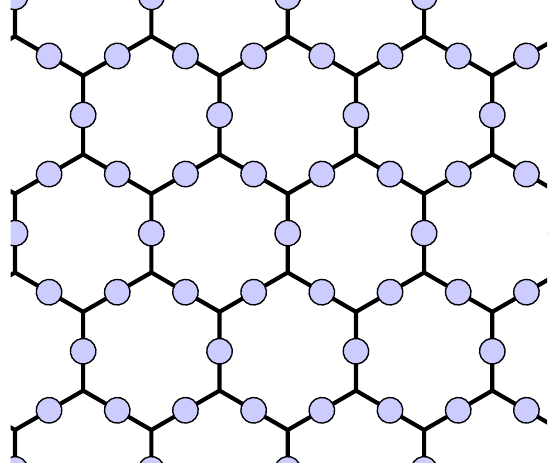}
		}
		\hspace{0.4cm}
		\subfloat[]{\label{F:DSdoublelattice}
		\includegraphics{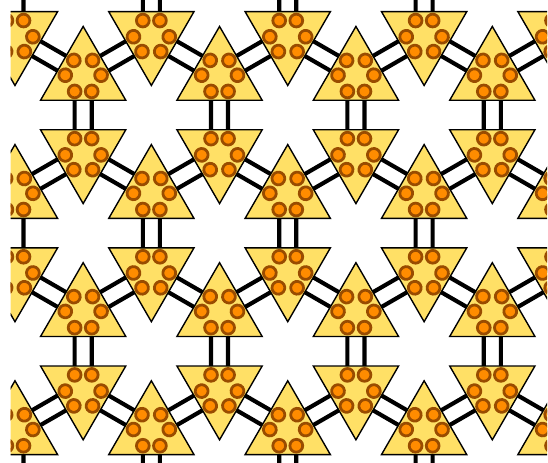}
		}
		\hspace{0.2cm}
		\subfloat[]{\label{F:DSvertexA}
		\includegraphics{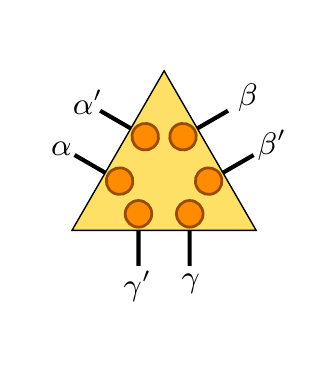}
		}
		\caption{(a) The double semion model is typically presented on a honeycomb lattice with qubits (blue) on links. (b) The PEPS representation of the same region of the lattice. At vertices of the honeycomb lattice, we place 6 virtual (orange) qubits, which will be projected into a $4$-dimensional (yellow) code qudit. Solid lines here denote edges of the PEPS graph, and so connect virtual qubits in maximally entangled states. (c) Virtual qubit labels. For triangles oriented in the opposite direction, rotate these definitions by $\pi$. States are labelled as $\ket{\alpha\alpha';\beta\beta';\gamma\gamma'}$.}
		\label{F:DSlattice}
		\end{figure}
		
		%------------------------------------------------------------------------------------------------------------%

		The double semion state is typically defined on a honeycomb lattice with qubits on the edges, as in \Fref{F:DShexlattice}. It is conventionally defined as the ground state of the Hamiltonian~\cite{LevinSN05}
		\begin{align}\label{e:dshamdefine}
			H_{\mathrm{ds}}=-\sum_v\prod_{j\sim v}\sigma^z_j+\sum_p\left(\prod_{k\in p}\sigma^x_k\right)\left(\prod_{m\sim p}i^{\frac{1-\sigma^z_m}{2}}\right)
		\end{align}
		where $v$ ($p$) are vertices (plaquettes) of the honeycomb lattice, $\sigma$ are the Pauli matrices, $j\sim v$ runs over qubits incident to vertex $v$, $k\in p$ runs over edges bounding plaquette $p$, and $m\sim p$ runs over edges incident to $p$ (i.e.~edges sharing one vertex with $p$). 
		
		We can represent the ground state of this Hamiltonian as a PEPS by placing 2 pairs of maximally entangled pairs of qubits between each vertex of the honeycomb lattice, as shown in \Fref{F:DSdoublelattice}, and applying a projection map $\pmap_s$ at each vertex to map from the $2^6$-dimensional virtual space to a $4$-dimensional \physicalspaceend.  (Note this is a slightly different PEPS representation for the double semion ground state as compared to those presented in Refs.~\cite{Gu2008, Gu-09, Buerschaper-09}.) The correspondence between these \physicalqudits and the qubits of the double semion model is not obvious at this stage, but will become clear as we proceed. At this point we should emphasise the distinction between the honeycomb lattice, on which the double semion model is typically defined (\Fref{F:DShexlattice}), and the PEPS lattice, whose edges correspond to maximally entangled virtual pairs (\Fref{F:DSdoublelattice}). In particular, we stress that the sites of the PEPS lattice (where $\pmap_s$ is applied) correspond to vertices of the honeycomb lattice, and not edges.
		
		For the most part of this analysis, we will neglect normalization for the sake of clarity. With this in mind, and the labelling conventions of \Fref{F:DSvertexA}, we can write the projection map $\pmap_s$ as:
		\begin{align}
			\pmap_s &= \sum_{\alpha,\beta,\gamma,i,j,k \in \Z_2} T_{\alpha\beta\gamma}\cdot\delta_{i=\alpha+\beta}\delta_{j=\beta+\gamma}\delta_{k=\gamma+\alpha} \ket{ijk}_{\p}\bra{\alpha\beta;\beta\gamma;\gamma\alpha}_{s}
		\end{align}
		where the $\p$ subscript on the ket indicates that it is a \physicalspace state, addition is modulo 2, and we have defined
		\begin{align}
			T_{\alpha\beta\gamma}&= \left\{\begin{array}{cc}1&\qquad\alpha+\beta+\gamma = 0,3\\i&\qquad\alpha+\beta+\gamma = 1\\-i&\qquad\alpha+\beta+\gamma = 2\end{array}\right.
		\end{align}		
		
		It should be clear that although there are in principle $8$ states that could be labelled by $i$, $j$, and $k$, these variables are not independent. In fact, there are only $4$ non-vanishing states of this form, given explicitly by
		\begin{align}
			\pmapbrac{}^{\dagger}\ket{000}_{\p} &= \ket{00;00;00}+\ket{11;11;11}\label{e:ijkstate1}\\
			\pmapbrac{}^{\dagger}\ket{110}_{\p} &= i\ket{10;01;11}-i\ket{01;10;00}\\
			\pmapbrac{}^{\dagger}\ket{101}_{\p} &= i\ket{01;11;10}-i\ket{10;00;01}\\
			\pmapbrac{}^{\dagger}\ket{011}_{\p} &= i\ket{11;10;01}-i\ket{00;01;10}\label{e:ijkstate4}
		\end{align}
		
		We call these values of $\{ijk\}$ the \textit{allowed} values. The site projector is then $P_s = \sum_{\{ijk\}}\pmapbrac{s}^{\dagger}\proj{ijk}_{\p,s}\pmap_s$, where the sum only runs over allowed values of $\{ijk\}$.
		
		The reason we use this redundant description of these states is that it allows us to identify the states $i$, $j$, and $k$ as the states of the qubits of the double semion model. That is, each variable is associated with an edge of the honeycomb model. That some states (e.g. $\ket{100}_p$) are not in the image of $\pmap_s$ is a consequence of the fact that these states do not belong to the double semion ground state, which is what this PEPS describes. One might also worry that we have two variables labelling the state of each edge (one for each vertex on which the edge ends). However, we will see that this is resolved in our analysis.
		
		Our construction proceeds by simulating the projection maps with code gadgets, using a Hamiltonian of the form of Eq.~(\ref{e:fullHami}).  Recall that the Hamiltonian for a code gadget is simply $Q_s =1-P_s$.  Because our virtual systems are qubits, each edge of the PEPS lattice has an associated operator
		\begin{align}
			M_e &= \proj{00}_e+\proj{11}_e+\ketbra{00}{11}_e+\ketbra{11}{00}_e
		\end{align}
		
		The full Hamiltonian of our system is then given by
		\begin{align}
			H&=\sum_sQ_s-\pert\sum_eM_e
		\end{align}
		where $s$ ($e$) runs over the sites (edges) of the PEPS lattice.

	\subsection{Effective Hamiltonian}
			
		In this example, we will not use the more rigorous global Schrieffer-Wolff perturbation method as in Sections \ref{s:gsw} and \ref{s:gsgsw}. We instead use the simpler self-energy expansion as used in Refs.~\cite{Kitaev-06, Brell-2011}. This amounts to neglecting the $\Gamma(q_1,\ldots,q_j)$ terms in the global SW expansion with any $|q_i|\neq 1$.
		
		Given the Hamiltonian
		\begin{equation}
			H=H_0+\pert V
		\end{equation}
		the self-energy low energy effective Hamiltonian is given by 
		\begin{align}\label{e:effHamSE}
			H_{\mathrm{eff}, \mathrm{SE}} &= E_0 + \sum_n \pert^n \pr_0 V\left(\frac{Q_0}{H_0}V\right)^{(n-1)}\pr_0
		\end{align}
		where $\pr_0$ is the projector to the ground space of $H_0$ with energy $0$, and $\frac{Q_0}{H_0}$ is defined to vanish on ground states of $H_0$. In writing the effective Hamiltonian in this way, we have neglected the dependence of the ground state energy on the perturbation. For our purposes, $O(1)$ constants are unimportant, so we will commonly neglect them in our analysis.			
			
		If we now explicitly evaluate the expansion (\ref{e:effHamSE}) for the double semion model, we will see that the terms arising will  provide a parent Hamiltonian for the desired state. We have for the low energy effective Hamiltonian
		\begin{align}\label{e:effHamDS}
			H_{\mathrm{eff},\mathrm{ds}} &= -\sum_n \pert^n \pr_0 \left(\sum_eM_e\right)\left(\frac{Q_0}{H_0}\left(\sum_eM_e\right)\right)^{(n-1)}\pr_0
		\end{align}
		
		We now evaluate this sum order by order. At $0\th$ order, the effective Hamiltonian can simply be taken to be
		\begin{align}
			H_{\mathrm{eff},\mathrm{ds}}^{\expect{0}} &= -\pr_0
		\end{align}
		This term may seem trivial, but in fact it is not so if we consider it in terms of the qubits in the double semion model. This term enforces the constraint that only allowed values of $\ket{ijk}_p$ at each vertex are in the ground space. Thinking about these variables $i$, $j$, and $k$ as labels for the states of the three qubits on the edges incident to any given vertex, this constraint plays the same role as the term $-\prod_{j\sim v}\sigma^z_j$ in Eq.~(\ref{e:dshamdefine}).

	At $1\st$ order, we find that the only terms to appear are also proportional to $\pr_0$, and so we can absorb them into constants of the $0\th$ order Hamiltonian. At $2\nd$ order, non-trivial terms can appear corresponding to each edge of the honeycomb lattice. Neglecting $O(1)$ constants, the effective Hamiltonian will take the form
		\begin{align}\label{e:ds2order}
			H_{\mathrm{eff},\mathrm{ds}}^{\expect{2}} &\sim -\pr_0-\pert^2\sum_{\substack{s\\s'\sim s}}\pr_0 C_{s,s'}\pr_0\\
			C_{s,s'}&\equiv\sum_{\substack{i,j,k\\i',j',k'}}\pmapbrac{s,s'}^{\dagger}\delta_{i'j'k'}^{ijk}(s,s')\proj{ijk}_{s}\cdot \proj{i'j'k'}_{s'}\pmap_{s,s'}
		\end{align}
		where $s'\sim s$ runs over all $s'$ neighbouring $s$ (with $s$ and $s'$ sites of the PEPS lattice). The function $\delta_{i'j'k'}^{ijk}(s,s')$ takes the form of a Kronecker delta $\delta_{ii'}$, $\delta_{jj'}$, or $\delta_{kk'}$ for $s$ and $s'$ connected by an edge running northeast, northwest, or vertically respectively. These second order terms arise from the product of $M_e$ on the two PEPS edges connecting $s$ and $s'$. Recalling that each site has its own label for the state of the double semion qubit on incident edges of the honeycomb lattice, the $C_{s,s'}$ terms can be interpreted as requiring that the two labels for the state (one each from $s$ and $s'$) are consistent. This resolves the apparent overcounting of degrees of freedom present in the model.
		
		If we continue expanding the effective Hamiltonian order by order, we will find that no new terms (that are not products of $2\nd$ order terms) arise until $6\th$ order. At this order, a new term will arise from the product of $M_e$ terms around the inside of a hexagonal plaquette. We can write the effective Hamiltonian (neglecting products of $2\nd$ order terms and constant factors) as
		\begin{align}\label{e:dseffham6}
			H_{\mathrm{eff},\mathrm{ds}}^{\expect{6}} &\sim -\pr_0-\pert^2\sum_{\substack{s\\s'\sim s}}\pr_0 C_{s,s'}\pr_0-\pert^6\sum_p \pr_0 B_p\pr_0
		\end{align}
		where the action of $B_p$ can be described as
	\begin{align}
	B_p\sim \prod_{e=(v,v')\in p}\left(\ketbra{00}{11}_{v,v'}+\ketbra{11}{00}_{v,v'}\right)
	\end{align}
	with $e=(v,v')$ the edges of the PEPS lattice comprising the interior of the plaquette $p$ (see \Fref{F:DSBp}).			
		
\begin{figure}
		\centering
		\subfloat{
		\includegraphics{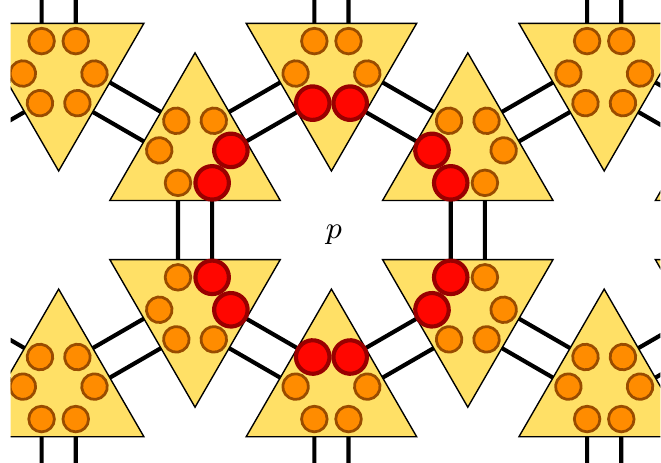}
		}
		\caption{The $B_p$ operator acts on the edges of the PEPS lattice closest to the centre of the plaquette $p$. The affected qubits are shown here in red.}
		\label{F:DSBp}
		\end{figure}

	We can more clearly examine the effect of $B_p$ by restricting to the image of all $P_0C_{s,s'}P_0$. Call the projector to this subspace $P_C$, and note that it is the ground space of $H_{\mathrm{eff},\mathrm{ds}}^{\expect{2}}$. Within this subspace, we can unambiguously assign \physicalspace state labels to the edges of the honeycomb lattice as in the standard definition of the double semion model. If we then evaluate $P_CB_pP_C$, we find that phases accumulate depending on the state of the edges leading out of the plaquette under consideration (the legs of the plaquette). This is due to the asymmetry between the phase factors defining the $\ket{ijk}_p$ states of Eqs.~(\ref{e:ijkstate1})-(\ref{e:ijkstate4}).
		
		On the double semion model states (i.e.~those associated with the honeycomb lattice), we can describe this by
		\begin{align}
			P_CB_pP_C &\sim %-\prod_{m\sim p}i^{\proj{1}_{m}}\prod_{k\in p} \left(\ketbra{0}{1}_k+\ketbra{1}{0}_k\right)
			\left(\prod_{k\in p}\sigma^x_k\right)\left(\prod_{m\sim p}i^{\frac{1-\sigma^z_m}{2}}\right)
		\end{align}
		where $k$ runs over honeycomb lattice edges comprising $p$ and $m$ runs over edges incident to $p$, precisely as in \Eref{e:dshamdefine}.
		
	Given also that the $P_0B_pP_0$ commute with the $P_0C_{s,s'}P_0$, this completes the specification of the ground space of this model. The effect of each type of term arising in the effective Hamiltonian on the low energy space can be summarized as follows.
			
		\begin{tabular}{p{0.165\columnwidth}p{0.796\columnwidth}}
		\qquad\textbf{$\mathbf{0^{\mathrm{\textbf{th}}}}$ order}: & Forbids disallowed vertex configurations (branching strings)\\
		\qquad\textbf{$\mathbf{2^{\mathrm{\textbf{nd}}}}$ order}: & Enforces consistency between the two descriptions of the state on each edge of the honeycomb lattice (only one qubit per edge)\\
		\qquad\textbf{$\mathbf{6^{\mathrm{\textbf{th}}}}$ order}: & Gives rise to plaquette energetics of the double semion model		
		\end{tabular}
		
		Each of these types of term acts on a different characteristic energy scale, based on the order of perturbation theory at which they arise. This gives the spectrum of our system as in \Fref{F:DSenergetic}. The ground state can easily be identified as (an encoded form of) the double semion ground state.
		
\begin{figure}
		\centering
		\subfloat{
		\includegraphics{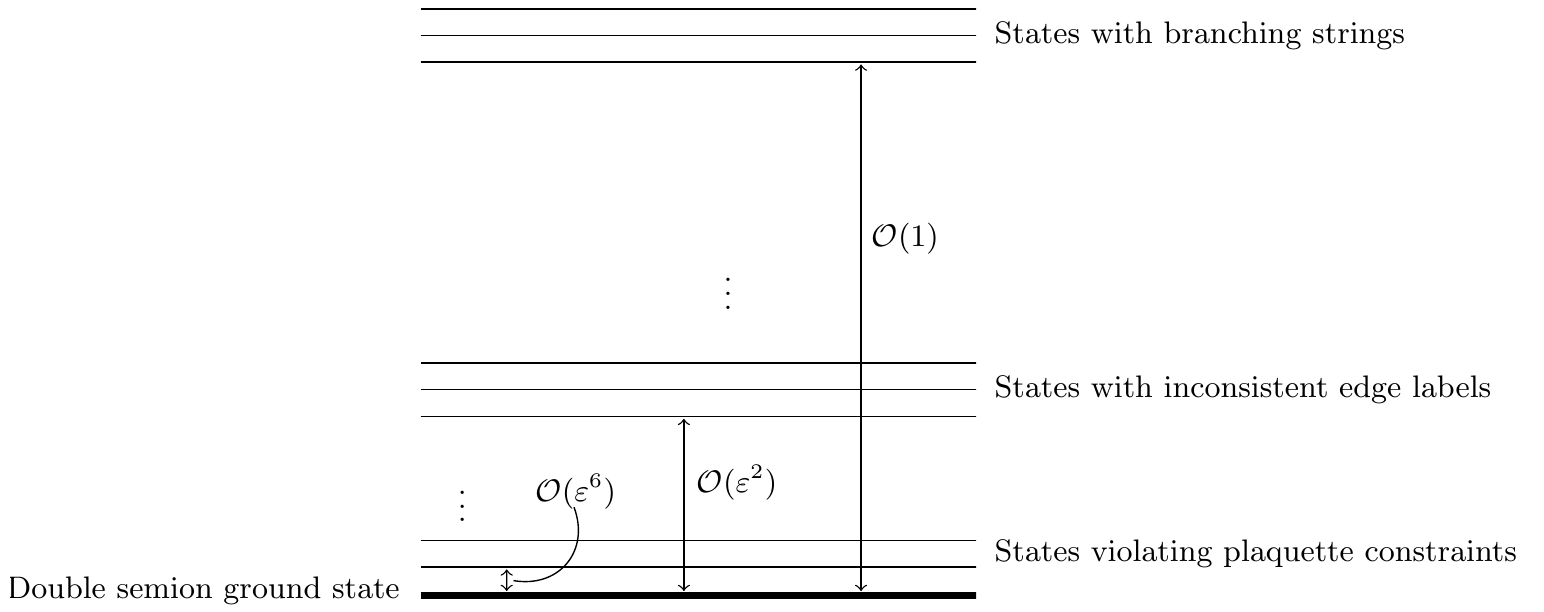}
		}
		\caption{A schematic of the energy level structure of our effective Hamiltonian for the double semion model defined by Eq.~(\ref{e:dseffham6}), showing the lowest excited states for each characteristic energy scale (each order of perturbation theory).}
		\label{F:DSenergetic}
		\end{figure}

		In previous examples of this kind of construction~\cite{Bartlett-2006,Brell-2011}, the local symmetries of the target model were exact symmetries of the full Hamiltonian (including perturbation). However, as discussed in \Sref{s:symmetry}, this is not a general feature of our construction. In this double semion example, note that the $2\nd$ order terms correspond to exact symmetries of the model, while the $6\th$ order terms do not.
		
		The Hamiltonian we presented is 2-local if a code gadget is considered as one system. However, if we consider the virtual qubits to be distinct particles, we would need to use further perturbation gadgets techniques as outlined in \Sref{s:reducecomplex} to reduce the Hamiltonian interactions to 2-local.

%------------------------------------------------------------------------------------------------------------%
%------------------------------------------------------------------------------------------------------------%

\begin{acknowledgments}
	Finally, we thank Steve Flammia, Ilai Schwarz, and Dominic Williamson for interesting and fruitful discussions.  This work is supported by the ARC via the Centre of Excellence in Engineered Quantum Systems (EQuS) project number CE110001013, by the ERC grant QFTCMPS, and by the cluster of excellence EXC 201 Quantum Engineering and Space-Time Research.
\end{acknowledgments}

%------------------------------------------------------------------------------------------------------------%

\end{document}